\documentclass[11pt]{article}%
\usepackage{amsmath}
\usepackage{amsfonts}
\usepackage{enumerate}
\usepackage{amssymb}
\usepackage{caption}
\usepackage{color}
\usepackage[colorlinks,linkcolor=red,anchorcolor=blue,citecolor=blue]%
{hyperref}
\usepackage{graphicx}
\usepackage{subcaption}
\usepackage{multirow}
\usepackage{appendix}
\usepackage{rotating}
\usepackage[round]{natbib}%
\providecommand{\U}[1]{\protect \rule{.1in}{.1in}}

\newtheorem{theorem}{Theorem}[section]

\newtheorem{assumption}{Assumption}[section]

\newtheorem{definition}{Definition}[section]
\newtheorem{example}{Example}[section]

\newtheorem{lemma}{Lemma}[section]

\newtheorem{remark}{Remark}[section]

\newenvironment{proof}[1][Proof]{\noindent \textbf{#1.} }{\  \rule{0.5em}{0.5em}}
\numberwithin{equation}{section}
\begin{document}

\title{Portfolio Credit Risk with Archimedean Copulas: Asymptotic Analysis and
Efficient Simulation}
\author{Hengxin Cui\thanks{\textrm{Manulife Financial Corporation, Canada. E-mail:
nordcui@gmail.com}}\quad Ken Seng Tan\thanks{\textrm{Division of Banking \&
Finance, Nanyang Business School, Nanyang Technological University, Singapore.
Email: kenseng.tan@ntu.edu.sg. ORCID Identifier: 0000-0002-2068-7829} }\quad Fan Yang\thanks{\textrm{Department of
Statistics and Actuarial Science, University of Waterloo, Canada. E-mail:
fan.yang@uwaterloo.ca}}}
\date{April 20, 2022}
\maketitle

\begin{abstract}
In this paper, we study large losses arising from defaults of a credit
portfolio. We assume that the portfolio dependence structure is modelled by
the Archimedean copula family as opposed to the widely used Gaussian copula.
The resulting model is new, and it has the capability of capturing extremal
dependence among obligors. We first derive sharp asymptotics for the tail
probability of portfolio losses and the expected shortfall. Then we
demonstrate how to utilize these asymptotic results to produce two variance
reduction algorithms that significantly enhance the classical Monte Carlo
methods. Moreover, we show that the estimator based on the proposed two-step
importance sampling method is logarithmically efficient while the estimator
based on the conditional Monte Carlo method has bounded relative error as the
number of obligors tends to infinity. Extensive simulation studies are
conducted to highlight the efficiency of our proposed algorithms for
estimating portfolio credit risk. In particular, the variance reduction
achieved by the proposed conditional Monte Carlo method, relative to the crude
Monte Carlo method, is in the order of millions.


\end{abstract}

\section{Introduction}

Credit risk in banking and trading book is one of the largest financial risk
exposures for many financial institutions.
As such quantifying the risk of a credit portfolio is essential in credit risk management. One of the key challenges in credit risk management is the accurate modelling of dependence between obligors, particularly in the tails. This is attributed to two phenomena. First, many financial institutions' exposure to credit risk  do not just confine to a single obligor, but to a large portfolio of multiple obligors. Second, empirically we have been observing simultaneous defaults in large credit portfolios as financial institutions are affected by common macroeconomic or systemic factors. This suggests that obligors tend to exhibit stronger dependence in a stressed market and hence simultaneous defaults tend to be more likely.  For this reason,  the model for the
dependence structure of the default events has a direct impact on the tail of
the loss for a large portfolio. In light of these issues and challenges, the first objective of this paper is to analyze the large credit portfolio
loss when the obligors are modeled to be strongly dependent. From the asymptotic analysis, the second objective of the paper is to propose two variance reduction simulation algorithms that provide efficient estimation of the risk of a credit
portfolio.

To accomplish the above two objectives, we model the credit risk based on the so-called threshold models which are widely used to capture the event of default for an
individual obligor within the portfolio. A default in a threshold model occurs if
some critical random variable, usually called a latent variable, exceeds (or
falls below) a pre-specified threshold. The dependence among defaults then
stems from the dependence among the latent variables. It has been found that
the copula representation is a useful tool for studying the dependence
structure. Specifically, the copula of the latent variables determines the
link between marginal default probabilities for individual obligors and joint
default probabilities for groups of obligors. Most threshold models used in the
industry are based explicitly or implicitly on the Gaussian copula; see for
example, CreditMetrics \citep{gupton1997creditmetrics} and Moody's KMV system
\citep{kealhofer2001portfolio}. While the Gaussian copula models
can accommodate a wide range of  correlation structure, they are 
inadequate to model extremal dependence between the latent variables as
these models are known to exhibit weaker dependence in the tails of the
underlying random variables (see for example, Section 7.2.4 of \citealp{mcneil2015quantitative} for further discussion on tail dependence).  This limitation raises considerable concern as we have already noted earlier that when the market
condition worsens, simultaneous defaults can occur with nonnegligible
probability in large credit portfolios.   {To better reflect the empirical evidence, copulas that
can capture \textquotedblleft stronger\textquotedblright \ tail dependence of
obligors such as $t$-copula and its generalizations have been proposed
\citep[see][]{bassamboo2008portfolio,chan2010efficient,tang2019sharp}. As pointed out in Section 11.1.4 of
\cite{mcneil2015quantitative}, the Archimedean copula is another plausible class of dependence model for obligors.   The Archimedean copula offers great flexibility in modeling
dependence  as it is capable of covering
dependence structures ranging from independence to comonotonicity (the
perfect dependence). Typical examples of Archimedean copulas include
Clayton, Gumbel and Frank copulas. Because of their flexibility in dependence modeling,   the Archimedean
copulas have been applied in credit risks (\citealp{cherubini2004copula};
\citealp{hofert2010sampling}; \citealp{hofert2011cdo};
\citealp{naifar2011modelling}), insurance (\citealp{frees1998understanding};
\citealp{embrechts2001modelling}; \citealp{denuit2004bivariate};
\citealp{albrecher2011explicit}; \citealp{cossette2018dependent}), and many other
areas of applications such as   \cite{genest2007everything},
\cite{zhang2007bivariate}, and \cite{wang2003estimating}. See also 
\cite{charpentier2009tails}, \cite{hofert2013archimedean}, \cite{okhrin2013structure} and \cite{ZhuWangTan2016}
for higher dimensional applications of Archimedean copulas
and their generalizations (such as the hierarchical Archimedean copulas) in  finance and risk management.
Motivated by  their dependence modeling flexibility and their 
wide applications, this paper similarly uses the Archimedean copula  to model the dependence of obligors in order
 to account for market phenomenon of simultaneous defaults.
Moreover, as to be discussed in Section \ref{sslt}, when obligors are modeled with
Archimedean copulas, the threshold model can  {{similarly}} be understood as a one-factor
Bernoulli mixture model, which leads to great conveniences for asymptotic
analysis and simulations of large portfolio losses in the later sections.}

In terms of quantifying portfolio credit risk, the most popular measure is to study
the probability of large portfolio loss over a fixed time horizon, say, a
year \citep[see][]{glasserman2004tail, glasserman2007large, tang2019sharp}. 
The expected shortfall of large
portfolio loss, which has been found to be very useful in the risk management and the pricing of
credit instruments, is another important measure of credit risk. A general discussion on quantifying portfolio
credit risk can be found in \cite{hong2014estimating}. With the extremal
dependence being modelled by the Archimedean copula, there are no  analytic expressions
for the above two measures. This implies we need to rely on numerical methods to evaluate these measures. While the Monte Carlo (MC) simulation method is a popular alternate numerical tool, naive application of the method to evaluate these measures is  very inefficient since the event of defaults of high-quality
obligors is rare and the naive MC is notoriously known to be inefficient for rare-event applications. Hence, variance reduction techniques such as that based on the 
 importance sampling \citep{glasserman2005importance, bassamboo2008portfolio,glasserman2008fast} and the
conditional Monte Carlo method \citep{chan2010efficient} have been proposed to increase the efficiency of MC methods for estimating these measures. 

We now summarize the key contributions of the paper. we
study the large portfolio loss from the following two aspects.
By exploiting the threshold model and  using an Archimedean copula to
capture the dependence among latent variables, the first key contribution is to  derive sharp asymptotics for the two performance measures: the probability of large portfolio loss and the
expected shortfall. These results quantify the asymptotic behavior of the two
measures when the portfolio size is large and provide better understanding on
how dependence affects the large portfolio loss. While the effectiveness
of estimating the portfolio loss based on the asymptotic expansions may
deteriorate if the portfolio size is not sufficiently large enough, it is
still very useful as it provides a good foundation for the MC based algorithms that we are developing subsequently. In particular, the second key contribution of the paper is to exploit these asymptotic results and propose two efficient MC based methods for estimating the above two performance measures. More specifically,  the first one is a two-step full importance sampling algorithm that provides efficient estimators for both measures. Furthermore, we show that
the proposed estimator  of probability of large portfolio loss  is logarithmically efficient. The second algorithm, which is
based on the conditional Monte Carlo method, is shown to have
bounded relative error. Relative to the importance sampling method, the conditional Monte Carlo  algorithm has the advantage of its simplicity though it  
can only be used to estimate the probability of portfolio loss. Simulation
studies also show the better performance of the second algorithm than the
first one. Overall, both of them generate significant variance reductions when
compared to naive MC simulations.

The rest of the paper is organized as follows. We formulate our problem in
Section \ref{spf} and describe  Archimedean copula and regular variation  in
Section \ref{sp}. Main results are presented in Sections \ref{saa}, \ref{siss}
and \ref{scmcs}, with Section \ref{saa} derives the sharp asymptotics and the
latter two sections present our proposed efficient Monte Carlo algorithms and analyze
their performances. Through an extensive simulation study, Section \ref{snr}
provides further comparable analysis on the relative effectiveness of our proposed algorithms.
Proofs are relegated to Appendix.

\section{Problem Formulation\label{spf}}

Consider a large credit portfolio of $n$ obligors. Similar to
\cite{bassamboo2008portfolio}, we employ a static structural model for
portfolio loss by introducing latent variables $\{X_{1},\ldots,X_{n}\}$ so
that each obligor defaults if each latent variable $X_{i}$ exceeds some
pre-specified threshold $x_{i}$, for $i=1,2,...,n$. By denoting $c_{i}>0$ as the risk exposure at default that corresponds to obligor $i$, for $i=1,2,...,n$, the portfolio loss incurred from defaults
is given by
\begin{equation}
L_{n}=\sum_{i=1}^{n}c_{i}1_{\{X_{i}>x_{i}\}}, \label{first}%
\end{equation}
where $1_{A}$ is the indicator function of an event $A$. Such a threshold
model can first trace back to \cite{merton1974pricing}. Let $F_{i}$ and  $\overline{F}_{i}$ be, respectively, the marginal distribution function and marginal survival function of $X_{i}$. Then $F_{i}=1-\overline{F}_{i}$, for
$i=1,2,...,n$. In the  threshold model, $\overline{F}_{i}$ can be interpreted as the marginal default probability of obligor $i$ and we use  $p_{i}$ to denote it.

As pointed out in the last section, the dependence among the latent variables has a
direct impact on the tail of the loss for a large portfolio and their dependence structure is conveniently modelled via copulas. This is also 
highlighted in Lemma 11.2 of \cite{mcneil2015quantitative} that in a threshold
model the copula of the latent variables determines the link between marginal
default probabilities and portfolio default probabilities. To see this, let
$U_{i}=F_{i}(X_{i})$ and $p_i= \overline{F}_{i}$ for $i=1,\ldots,n$. It follows immediately from Lemma
11.2 of \cite{mcneil2015quantitative} that $(X_{i},x_{i})_{1\leq i\leq n}$ and
$(U_{i},p_{i})_{1\leq i\leq n}$ are two
equivalent threshold models. Then the portfolio loss is affected by the
dependence among the latent variables rather than the distribution of each
latent variable. This is also the reason why we conduct our analysis by focusing on the dependence structure of the obligors and with the  assumption that the dependence of $(U_{1},U_{2},...,U_{n})$ is adequately captured by an Archimedean copula.

Recall that the main focus of the paper is to study the credit portfolio for which the
portfolio consists of a large number of obligors and each obligor has low
default probability. While default events are rare, the potential loss is significant once they are triggered and with simultaneous defaults.   By the transformation $U_{i}=F_{i}(X_{i})$ for $i=1,\ldots,n$, the event
that obligor defaults $\{X_{i}>x_{i}\}$ is equivalent to $\{U_{i}>1-p_{i}\}$. From the theory of diversification, the
probability of large portfolio loss should diminish as $n$ increases. To
capture this feature, the individual default probability can also be expressed
as $p_{i}=l_{i}f_{n}$ for $i=1,\ldots,n$, where $f_{n}$ is a positive
deterministic function converging to $0$ as $n\rightarrow \infty$ and
$\{l_{1},\ldots,l_{n}\}$ are strictly positive constants accounting for
variations effect on different obligors. {We emphasize that on one hand, the assumption that $f_{n}$ converges to $0$ is to reflect the  diversification effect in a large portfolio that individual default probability diminishes as $n$ increases. On
the other hand, it provides mathematical convenience to derive sharp
asymptotics for the large portfolio loss (see discussions
after Theorem \ref{th4.1}) and to prove the algorithms' efficiency (Theorems
\ref{coasyoptimal} and \ref{coboundederror}). Such condition is also assumed in \cite{gordy2003risk}, \cite{bassamboo2008portfolio}, \cite{chan2010efficient} and
\cite{tang2019sharp}, for example.} {More detailed explanations on the assumption of $f_n$ are provided in Section \ref{sec:tail}.} With this representation, we can rewrite the overall
portfolio loss (\ref{first}) as
\begin{equation}
L_{n}=\sum_{i=1}^{n}c_{i}1_{\{U_{i}>1-l_{i}f_{n}\}}. \label{m1}%
\end{equation}
In the remaining of the paper, we use (\ref{m1}) to analyze the large
portfolio loss. To characterize the potential heterogeneity among obligors, we
further impose some restrictions on the sequence $\{(c_{i},l_{i}):i\geq1\}$,
as in \cite{bassamboo2008portfolio}.

\begin{assumption}
\label{ass1}Let the positive sequence $((c_{i},l_{i}):i\geq1)$ take values in
a finite set $\mathcal{W}$. By denoting $n_{j}$ as the number of each element
$(c_{j},l_{j})\in \mathcal{W}$ in the portfolio,  we further assume that $n_{j}/n$
converges to $w_{j}>0$, for each $j\leq|\mathcal{W}|$ as $n\rightarrow \infty$.
\end{assumption}

In practice, Assumption \ref{ass1} can be interpreted as a heterogeneous
credit portfolio that comprises of a finite number of homogeneous
sub-portfolios based on risk types and exposure sizes. We note that it is easy
to relax this assumption to the case where $c_{i}$ and $l_{i}$ are random
variables; see \cite{tong2016exposure} and \cite{tang2019sharp} for recent discussions.

\section{Preliminaries}

\label{sp}

\subsection{Archimedean copulas\label{sslt}}

Archimedean copulas have a simple closed form and can be represented by a
generator function $\phi$ as follows:
\begin{equation}
C(u_{1},\ldots,u_{n})=\phi^{-1}(\phi(u_{1})+\ldots+\phi(u_{n})), \label{arch}%
\end{equation}
where $C:[0,1]^{n}\rightarrow \lbrack0,1]$ is a copula function. The generator
function $\phi:[0,1]\rightarrow \lbrack0,\infty]$ is continuous, decreasing and
convex such that $\phi(1)=0$ and $\phi(0)=\infty$, and $\phi^{-1}$ is the
inverse of $\phi$. We further assume $\phi^{-1}$ is completely monotonic, i.e.
$(-1)^{i}\left(  \phi^{-1}\right)  ^{(i)}\geq0$ for all $i\in \mathbb{N}$,
which allows $\phi^{-1}$ to be a Laplace-Stieltjes (LS) transform of a
distribution function $G$ on $[0,\infty]$ such that $G(0)=0$. Let $V$ be a
random variable with a distribution function $G$ on $[0,\infty]$. The LS
transform of $V$ (or $G$) is defined as
\[
\phi^{-1}(s)=\mathcal{L}_{V}(s)=\int_{0}^{\infty}e^{-sv}\mathrm{d}%
G(v)=\mathbb{E}\left[  e^{-sV}\right]  ,\qquad s\geq0.
\]
Archimedean copulas that are generated from LS transforms of different
distributions are referred as LT-Archimedean copulas, as formally
defined below:
\begin{definition}
\label{def1}An LT-Archimedean copula is a copula of the form \eqref{arch},
where $\phi^{-1}$ is the Laplace-Stieltjes transform of a distribution
function $G$ on $[0,\infty]$ such that $G(0)=0$.
\end{definition}

For many popular Archimedean copulas, the random variable $V$ has a known
distribution. For example, $V$ is Gamma distributed for Clayton copulas, while
 $V$ is a one-sided stable random variable for Gumbel copulas. A detailed
specification on $V$ can be found in Table 1 of
\cite{hofert2008sampling}.

The following result provides a stochastic representation of $\mathbf{U}%
=(U_{1},\ldots,U_{n})$ where $\mathbf{U}$ follows an LT-Archimedean copula of the form:
\begin{equation} 
\mathbf{U}=\left(  \phi^{-1}\left(  \frac{R_{1}}{V}\right)  ,\ldots,\phi
^{-1}\left(  \frac{R_{n}}{V}\right)  \right) . \label{representlt}%
\end{equation}
Here $V$ is a positive random variable with LS transform $\phi^{-1}$ and
$\{R_{1},\ldots,R_{n}\}$ is a sequence of independent and identically
distributed (i.i.d.) standard exponential random variables independent of $V$.
This representation is first recognized by \cite{marshall1988families} and
later is formally proved in Proposition 7.51 of \cite{mcneil2015quantitative}.

The construction \eqref{representlt} is especially useful in the field of credit risk. To
see this, let us consider the threshold model defined in \eqref{m1} and that
$\mathbf{U}$ has an LT-Archimedean copula with generator $\phi$ as defined in \eqref{representlt}. Then the
random variable $V$ can be considered as a proxy for systematic risks.
Conditioning on $V$, random variables $U_{1},\ldots,U_{n}$ are independent
with conditional distribution function $\mathbb{P}(U_{i}\leq u|V=v)=\exp
(-v\phi(u))=p_{i}(v)$ for some predetermined $u\in \lbrack0,1]$. {By such a
construction, the threshold model \eqref{m1} can be represented succinctly as
a one-factor Bernoulli mixture model with mixing variable $V$ and mixing
probabilities $p_{i}(v),i=1,\ldots,n$. This property offers two important aspects. 
First is that it facilitates the asymptotic analysis of large portfolio loss. By understanding it as a one-factor
Bernoulli mixture model, we will show in Section \ref{saa} that the large portfolio loss is essentially determined by
the mixing distribution of $V$ or its LS transform $\phi^{-1}$. This is also
observed in \cite{gordy2003risk} that asymptotic analysis provides a simple
yet quite accurate way of evaluating large portfolio loss. In the current
paper, we push one step further by deriving the sharp asymptotics
for the large portfolio loss in a more explicit way and under the
Archimedean copula model. Second,  the Bernoulli mixture models lend themselves to
practical implementation of MC simulations \citep[see][]{mcneil2015quantitative, BHS2018}. To be more specific,  a Bernoulli mixture model can be simulated by first generating a realization $v$ of $V$ and then conducting independent Bernoulli experiments with conditional default
probabilities $p_{i}(v)$. This generation algorithm is explicitly exploited in Section \ref{siss} as the starting point of our proposed importance sampling simulations.}

\subsection{Regular variation}

Regular variation is an important notion in our modeling. Intuitively, a
function $f$ is regularly varying at infinity if it behaves like a power law
function near infinity. Interested readers may refer to
\cite{bingham1989regular} and \cite{resnick2013extreme} for textbook
treatments. In our model, we assume the generator function $\phi$ of the
LT-Archimedean copula is regularly varying in order to capture the upper tail dependence.

\begin{definition}
A positive Lebesgue measurable function $f$ on $(0,\infty)$ is said to be
regularly varying at $\infty$ with index $\alpha \in \mathbb{R}$, written as
$f\in \mathrm{RV}_{\alpha}$, if for $x>0$,
\[
\lim_{t\rightarrow \infty}\frac{f(tx)}{f(t)}=x^{\alpha}.
\]

\end{definition}

Similarly, $f$ is said to be regularly varying at $0$ if $f(\frac{1}{\cdot
})\in \mathrm{RV}_{\alpha}$ and $f$ to be regularly varying at $a>0$ if
$f(a-\frac{1}{\cdot})\in \mathrm{RV}_{\alpha}$.

It turns out that many LT-Archimedean copulas which are commonly used in practice have generators that are regularly varying at $1$. For example,
Gumbel copula has a generator of $\phi(t)=(-\ln(t))^{\alpha}$ for $\alpha
\in \lbrack1,\infty)$, and it follows that $\phi^{-1}$ is completely monotonic
and $\phi(1-\frac{1}{\cdot})\in \mathrm{RV}_{-\alpha}$.

\section{Asymptotic Analysis for Large Portfolio Loss\label{saa}}

In this section, we conduct an asymptotic analysis on a regime where the
number of obligors is large with each individual obligor having an excellent credit
rating (i.e. with small default probability). Our asymptotic analysis  focuses on large portfolio losses with Subsection~\ref{sec:tail}   analyzes the tail probabilities of the losses and   Subsection~\ref{sec:ES} tackles the
expected shortfall of the losses.

\subsection{Asymptotics for probabilities of large portfolio loss} \label{sec:tail}

{By considering the portfolio loss model \eqref{m1}, this subsection  analyzes the asymptotic probability  $\mathbb{P}(L_{n}>nb)$ as $n\rightarrow \infty$,
where $b$ is an arbitrarily fixed number.  We restrict our analysis
to LT-Archimedean copulas for modeling the dependence among the latent variables
in order to fully take advantage of the Bernoulli mixture structure (as
explained in Section \ref{sslt}). Recall the random variable $V$ in the representation
(\ref{representlt}) can be interpreted as the systematic risk  or
 common shock factor. The dependence of obligors is mainly induced by $V$. Note
that $\phi^{-1}$ is a decreasing function. When $V$ takes on large values, all
$U_{i}$'s are likely to be large (close to $1$), which leads to many obligors
default simultaneously. To incorporate strong dependence among obligors, we
assume $V$ to be heavy tailed. In particular, we assume $\overline{F}_{V}%
\in \mathrm{RV}_{-1/\alpha}$, with $\overline{F}_{V}(\cdot)$ corresponds to the survival
distribution function of $V$ and $\alpha>1$. Here $\alpha$ represents the
heavy tailedness of $V$ so that the larger $\alpha$ is, the heavier tail $V$ has
and the more dependent the obligors are, which means simultaneous defaults are
more likely to occur. By  Karamata's Tauberian theorem, this is equivalent to
assuming that $\phi(1-\frac{1}{\cdot})\in \mathrm{RV}_{-\alpha}$ with $\alpha>1$,
where by the convexity of the generator $\phi$ the condition $\alpha>1$
necessarily holds. Such heavy tailed assumption on the systematic risk factor
(or common shock) can be seen in \cite{bassamboo2008portfolio},
\cite{chan2010efficient} and \cite{tang2019sharp}. This is formalized in the
following assumption.}

\begin{assumption}
\label{A2}Assume $\mathbf{U}=(U_{1},\ldots,U_{n})$ follows an LT-Archimedean
copula with generator $\phi$ satisfying that $\phi(1-\frac{1}{\cdot}%
)\in \mathrm{RV}_{-\alpha}$ with $\alpha>1$. Let $V$ be the random variable
associated with Laplace-Stieltjes transform $\phi^{-1}$. Assume that $V$ has a
eventually monotone density function.
\end{assumption}

Before presenting the main result of this section, it is useful to note that
by conditioning on $V=\dfrac{v}{\phi(1-f_{n})}$, we have
\begin{align}
p(v,i) &  :=\mathbb{P}\left(  U_{i}>1-l_{i}f_{n}\Big \vert V=\frac{v}%
{\phi(1-f_{n})}\right)  \nonumber \\
&  =1-\exp \left(  -v\frac{\phi(1-l_{i}f_{n})}{\phi(1-f_{n})}\right)
.\label{p1}%
\end{align}
Under Assumption \ref{A2} that $\phi(1-\frac{1}{\cdot})\in \mathrm{RV}%
_{-\alpha}$, we immediately obtain
\[
\lim_{n\rightarrow \infty}p(v,i)=1-\exp \left(  -vl_{i}^{\alpha}\right)
:=\tilde{p}(v,i).
\]
With the condition $V=\dfrac{v}{\phi(1-f_{n})}$, and by the Kolmogorov's strong
law of large numbers, it follows that, almost surely
\begin{equation}
\frac{L_{n}|V=\frac{v}{\phi(1-f_{n})}}{n}\rightarrow r(v):=\sum_{j\leq
|\mathcal{W}|}c_{j}w_{j}\tilde{p}(v,j),\qquad \text{as }n\rightarrow
\infty.\label{r1}%
\end{equation}
Recall $w_j$ is formally defined in Assumption \ref{ass1}. 
Note that $r(v)$ is strictly increasing in $v$ and attains its upper bound
$\bar{c}=\sum_{j\leq|\mathcal{W}|}c_{j}w_{j}$ at infinity, where $\bar{c}$ can
be interpreted as the limiting average loss when all obligors default. Thus,
for each $b\in(0,\bar{c})$, we denote $v^{\ast}$ as the unique solution to
\begin{equation}
r(v)=b.\label{eq}%
\end{equation}
Essentially, $v^{\ast}$ represents the threshold value so that for
$V\in(0,v^{\ast}/\phi(1-f_{n}))$, the limiting average portfolio loss $\bar
{c}$ is less than $b$; for $V\in(v^{\ast}/\phi(1-f_{n}),\infty)$, the limiting
average portfolio loss $\bar{c}$ is greater than $b$.

Now we are ready to present the main theorem of this section which gives a
sharp asymptotic for the probability of large portfolio losses. The proof is
relegated to Appendix.

\begin{theorem}
\label{th4.1}Consider the portfolio loss defined in \eqref{m1}. Under
Assumptions \ref{ass1} and \ref{A2} and further assume that $\exp(-n\beta
)=o(f_{n})$ for any $\beta>0$. Then for any fixed $b\in(0,\bar{c})$, as
$n\rightarrow \infty$,%
\begin{equation}
\mathbb{P}(L_{n}>nb)\sim f_{n}\frac{(v^{\ast})^{-1/\alpha}}{\Gamma
(1-1/\alpha)}, \label{sa1}%
\end{equation}
where $v^{\ast}$ is the unique solution that solves (\ref{eq}).
\end{theorem}

{\begin{remark}
We emphasize that the asymptotic behavior of the portfolio loss is mostly dictated by $\alpha$
and $f_{n}$. Recall $\alpha$ is the index of regular variation of the
generator function $\phi$. It controls dependence among the latent variables:
the larger $\alpha$, the more likely that obligors tend to default
simultaneously. Once $\alpha$ is fixed, Theorem \ref{th4.1} shows that the
probability of large portfolio loss diminishes to zero at the same rate as
$f_{n}$. This result is sharp and is more explicit, compared to the results in
\cite{gordy2003risk}. As explained in Subsection \ref{sslt}, the threshold model \eqref{m1} under the Archimedean copula is reduced to a Bernoulli mixture model while \cite{gordy2003risk} studies a risk-factor model, which is essentially a Bernoulli mixture model.  In that paper, the author shows that
the capital of the fine-grained portfolio asymptotically converges to the
capital of the systematic risk factor. In other words, the asymptotics are presented between the
portfolio loss and the systematic risk factor (equivalent to the random
variable $V$ here). Moreover, through numerical experiments, the author shows the approximation of the
large portfolio loss by the systemic risk factor is quite accurate and simple.
In our current paper, we are able to go one step further that the large
portfolio loss is given in a more explicit form that it is linearly proportional
to the individual default probabilities $f_{n}$. Thus, to
approximate the large portfolio loss, we do not rely on the full information
of the systematic risk factor. At the same time, the asymptotics derived in
our paper should share the same advantages of accuracy and simplicity as discussed in
\cite{gordy2003risk}.
\end{remark}

{Now we discuss the assumption on $f_n$ in greater details. As mentioned earlier, due to the effect of diversification, individual default
probability diminishes in a large portfolio as $n$ increases. On the technical
side, letting $f_{n}$ converge to $0$ ensures that a large portfolio loss
occurs primarily when $V$ takes large values, whereas $R_{i}$, $i=1,\ldots,n$,
generally does not play a role in the occurrence of the large portfolio loss.
To better understand this requirement, we consider the case with $f_{n}\equiv
f$ being a constant. Then similar calculations as in (\ref{p1}) leads to
\[
p_{0}(v,i):=\mathbb{P}\left(  U_{i}>1-l_{i}f|V=v\right)  =1-\exp
(-v\phi(1-l_{i}f)).
\]
Note that $p_{0}(v,i)$ is strictly increasing in $v$. Under the condition
$V=v$, and the Kolmogorov's strong law of large numbers, we have, almost
surely,
\[
\frac{L_{n}|V=v}{n}\rightarrow r_{0}(v):=\sum_{j\leq|\mathcal{W}|}c_{j}%
w_{j}p_{0}(v,j),\qquad \text{as }n\rightarrow \infty,
\]
where the limit follows from Assumption \ref{ass1}. Clearly, $r_{0}(v)$ is
also strictly increasing in $v$. Define $v_{0}^{\ast}$ as the unique solution
to
\[
r_{0}(v)=b.
\]
It then follows for portfolio size $n$ large enough, we have $\mathbb{P}%
(L_{n}>nb|V=v)=0$ for $v\leq v_{0}^{\ast}$; and $\mathbb{P}(L_{n}>nb|V=v)=1$
for $v>v_{0}^{\ast}$. Thus, for any $b\in(0,\bar{c})$, and large enough $n$,
we have
\[
\mathbb{P}(L_{n}>nb)=\mathbb{E}[\mathbb{P}(L_{n}>nb|V)]=\overline{F}_{V}%
(v_{0}^{\ast}).
\]
This leads to a mathematically trivial result. Moreover, it is
counter-intuitive in the sense that as the size of the portfolio increases,
the probability of portfolio loss is still significant (i.e. not converging to
$0$ {which is in contradiction with} the portfolio diversification). {This
illustration exemplified the importance of the assumption that $f_{n}$
diminishes to $0$ as $n\rightarrow \infty$ to account for the rarity of large
loss. The assumption $\exp(-n\beta)=o(f_{n})$ essentially requires that the
decay rate of $f_{n}$ to $0$ needs to be slower than an exponential function.
By choosing different $f_{n}$, portfolios will have different credit rating
classes. For example, if $f_{n}$ decays at a faster rate such as $1/n$, then
the portfolio has higher quality obligors, whereas if $f_{n}$ decays at a
slower rate of $1/\ln n$, then the portfolio consists of more risky obligors.
There are also many similar discussions on the requirement of individual
default probability diminishes in a large portfolio (equivalent to our
$f_{n}\rightarrow0$) in the literature, which are all rooted in the effect of
diversification of a large portfolio. For example, in \cite{gordy2003risk}, the assumption (A-2)
}guarantees that the share of the largest single exposure in total portfolio
exposure vanishes to zero as the number of exposures in the portfolio
increases. \cite{tang2019sharp} explained it more explicitly in their Section 2 as follows. As
the size of the portfolio increases,  each latent variable $X_{i}$ should be
modified to $\frac{X_{i}}{\iota_{i}g_{n}}$, where $g_{n}$ is a positive
function diverging to $\infty$ to reflect an overall improvement on the credit
quality, and $\iota_{i}$ is a positive random variable to reflect a minor
variation in portfolio effect on obligor $i$. With the endogenously determined
default threshold of obligor $i$ is fixed as $a_{i}>0$, the individual default
occurs as $\frac{X_{i}}{\iota_{i}g_{n}}>a_{i}$ if and only if $X_{i}>\iota
_{i}a_{i}g_{n}$, which is equivalent to have $U_{i}<1-l_{i}f_{n}$ with $f_{n}$
decreasing to $0$ in our context.} \bigskip

Next we use an example involving a fully homogeneous portfolio to further
illustrate our results.

\begin{example}
\label{ex4.1}\normalfont Assume a fully homogeneous portfolio, that is
$l_{i}\equiv l,c_{i}\equiv c$. Under this assumption, \eqref{r1} simplifies
to
\[
r(v)=c\left(  1-\exp \left(  -vl^{\alpha}\right)  \right)  .
\]
Thus, $v^{\ast}=l^{-\alpha}\ln \dfrac{c}{c-b}$ is the unique solution to
$r(v)=b$. It immediately follows from relation \eqref{sa1} that, for
$b\in(0,c)$, we have
\begin{equation}
\mathbb{P}(L_{n}>nb)\sim lf_{n}\frac{\left(  \ln \frac{c}{c-b}\right)
^{-1/\alpha}}{\Gamma(1-1/\alpha)}. \label{example1}%
\end{equation}
Direct calculation further shows that the right-hand side of \eqref{example1}
is an increasing function of $\alpha$ if $\ln \frac{c}{c-b}\geq \exp(-\gamma)$,
i.e., $b/c\geq1-e^{-e^{-\gamma}}$, where $\gamma$ denotes the Euler's
constant. This monotonic result can be interpreted in an intuitive way. Recall
that $\alpha$ is the index of regular variation of the generator function
$\phi$. A larger $\alpha$ corresponds to a stronger upper tail dependence, and
therefore a joint default of obligors is more likely to occur. However, the
monotonicity fails if $b$ is not large. In this case, the mean portfolio loss
($L_{n}/n$) is compared to a lower level $b$ and such event may occur due to a
single obligor default. Thus, both the upper tail dependence and the level of
mean portfolio loss affect probability of large portfolio loss.
\end{example}

\subsection{Asymptotics for expected shortfall of large portfolio loss}\label{sec:ES}

The asymptotic expansions on the tail probabilities of the large portfolio
loss provide the foundation in the analysis of the expected shortfall. To see
this, the expected shortfall can be rewritten as
\begin{equation}
\mathbb{E}\left[  L_{n}|L_{n}>nb\right]  =nb+n\frac{\int_{b}^{\infty
}\mathbb{P}\left(  L_{n}>nx\right)  \mathrm{d}x}{\mathbb{P}\left(
L_{n}>nb\right)  }. \label{essurvival}%
\end{equation}
Theorem \ref{th4.1} becomes the key to establishing an asymptotic for the
expected shortfall, as formally stated in the following theorem.
\begin{theorem}
\label{th4.2}Under the same assumption as in Theorem \ref{th4.1}, the
following relation
\begin{equation}
\mathbb{E}\left[  L_{n}|L_{n}>nb\right]  \sim n\psi(\alpha,b) \label{essa}%
\end{equation}
holds for any fixed $b\in(0,\bar{c})$, where
\[
\psi(\alpha,b):=b+\frac{\int_{v^{\ast}}^{\infty}r^{\prime}(v)v^{-1/\alpha
}\mathrm{d}v}{(v^{\ast})^{-1/\alpha}}.
\]

\end{theorem}

The above theorem states that the expected shortfall grows almost linearly
with the size of the portfolio $n$.

\section{Importance Sampling (IS) Simulations for Large Portfolio Loss\label{siss}}

The asymptotic results established in the last section (see Theorems
\ref{th4.1} and \ref{th4.2}) characterize the behavior of large portfolio
losses. These results, however, may not be applicable in practical
applications unless the size of portfolio is large. In practice, the tail
probability or the expected shortfall of the portfolio loss are typically
estimated via MC simulation methods due to the non-tractability. Naive
application of MC method to this type of rare-event problems, on the other
hand, is notoriously inefficient. For this reason, variance reduction methods
are often used to enhance the underlying MC methods. We similarly follow this
line of inquiry and propose two variance reduction algorithms. In particular,
an IS algorithm based on a hazard rate twisting is
presented in this section while a second algorithm based on the conditional
Monte Carlo simulations will be discussed in the next section. The asymptotic
analysis in Section \ref{saa} plays an important role in proving the
efficiency of both algorithms.

\subsection{Preliminary of importance sampling}

We are interested in estimating $\mathbb{P}\left(  L_{n}>nb\right)  $, where
$L_{n}$ can be considered as a linear combination of conditionally independent
Bernoulli random variables $\{1_{\{U_{i}>1-l_{i}f_{n}\}},i=1,\ldots,n\}$. For
each Bernoulli variable, the associated probability is denoted by $p_{j}$ for
$j\leq|\mathcal{W}|$, which is a function of the generated variable $V$.
Following the analysis in Section \ref{saa}, $p_{j}$ is explicitly
given as $p(v,j)$ as shown in \eqref{p1}. The simulation of $\mathbb{P}\left(
L_{n}>nb\right)  $ is then conducted in two steps. In step 1, the common
factor $V$ using the density function $f_{V}(\cdot)$ is simulated and in step
2, the corresponding Bernoulli random variables are generated. When the
portfolio size is very large, the event $\{L_{n}>nb\}$ only occurs when $V$
takes large values and it further leads to the default probability $p_{j}$ for
each Bernoulli variable is small. Thus, both steps in the simulation of
$\mathbb{P}\left(  L_{n}>nb\right)  $ are rare event simulations. Estimation
by naive MC simulation becomes impractical due to the large number of samples
needed, and therefore, one has to resort to variance reduction techniques. IS
is a widely used variance reduction technique by placing greater probability
mass on the rare event of interest and then appropriately normalizing the
resulting output. Next, we briefly discuss how we apply IS in the two-step simulation.

In the first step, the tail behavior of large portfolio loss highly depends on
the tail distribution of $V$, i.e., the key to the occurrence of the large
loss event corresponds to $V$ taking large value. Then a good importance
sampling distribution for random variable $V$ should be more heavy-tailed than
its original distribution, so that a larger probability could be assigned to
the event that the average portfolio loss conditioned on $V$ exceeds the level
$b$. Such importance sampling distribution can be obtained via hazard rate
twisting on $V$. Let $\tilde{f}_{V}(\cdot)$ denote the importance sampling
density function for $V$ after the application of IS. In the second step, we improve the
efficiency of calculating the conditional probabilities by replacing each
Bernoulli success probability $p_{j}$ by some probability $\tilde{p}_{j}$, for
$j\leq|\mathcal{W}|$. In this case, exponentially twisting is a fairly
well-established approach for Bernoulli random variables; see, e.g.,
\cite{glasserman2005importance}. Let $\tilde{\mathbb{P}}$ denote the
corresponding IS probability measure and $\tilde{\mathbb{E}}$ be the
expectation under the measure $\tilde{\mathbb{P}}$. Then the following
identity holds:
\begin{equation} \label{eq:IS}
\mathbb{P}\left(  L_{n}>nb\right)  =\mathbb{E}\left[  1_{\{L_{n}>nb\}}\right]
=\tilde{\mathbb{E}}\left[  1_{\{L_{n}>nb\}}\tilde{L}\right]  ,
\end{equation}
where $\tilde{L}=\dfrac{d\mathbb{P}}{d\tilde{\mathbb{P}}}$ is the
Radon-Nikodym derivative of $\mathbb{P}$ with respect to $\tilde{\mathbb{P}}$
and equals
\[
\frac{f_{V}(V)}{\tilde{f}_{V}(V)}\prod_{j\leq|\mathcal{W}|}\left(  \frac
{p_{j}}{\tilde{p}_{j}}\right)  ^{n_{j}Y_{j}}\left(  \frac{1-p_{j}}{1-\tilde
{p}_{j}}\right)  ^{n_{j}(1-Y_{j})}.
\]
In the above expression,  $Y_{j}=1_{\{U_{j}>1-l_{j}f_{n}\}}$ and $n_{j}Y_{j}$ denotes the number
of defaults in sub-portfolio $j$. We refer to $\tilde{L}$ as the unbiasing
likelihood ratio.
The key finding from \eqref{eq:IS} is that calculating the tail probability $\mathbb{P}\left(  L_{n}>nb\right)$ is equivalent to evaluating either expectation $\mathbb{E}\left[  1_{\{L_{n}>nb\}}\right]$ or $\tilde{\mathbb{E}}\left[  1_{\{L_{n}>nb\}}\tilde{L}\right]$. We refer the estimator based on the latter expectation as the IS estimator and its efficiency crucially depends on the choice of the IS density function  $\tilde{f}_{V}(\cdot)$.

We now discuss two measures to characterize the performance of the proposed
IS estimator. 
Asymptotically,  the  good  performance  commonly  observed  in  realistic situations is a bounded relative error \citep[see][]{asmussen2006improved,mcleish2010bounded}. We say a sequence of estimators $(1_{\{L_{n}%
>nb\}}\tilde{L}:n\geq1)$ under probability measure $\tilde{\mathbb{P}}$ has
bounded relative error if
\[
\limsup_{n\rightarrow \infty}\frac{\sqrt{\tilde{\mathbb{E}}\left[
1_{\{L_{n}>nb\}}\tilde{L}^{2}\right]  }}{\mathbb{P}\left(  L_{n}>nb\right)
}<\infty.
\]
A slightly weaker form criterion called asymptotically optimal is also widely
used  \citep[see][]{glasserman2005importance, glasserman2007large,glasserman2008fast} if the following condition holds,
\[
\lim_{n\rightarrow \infty}\frac{\log \tilde{\mathbb{E}}\left[  1_{\{L_{n}%
>nb\}}\tilde{L}^{2}\right]  }{\log \mathbb{P}\left(  L_{n}>nb\right)  }=2.
\]
This condition is equivalent to saying that $\lim \limits_{n\rightarrow \infty
}\tilde{\mathbb{E}}\left[  1_{\{L_{n}>nb\}}\tilde{L}^{2}\right]
/\mathbb{P}\left(  L_{n}>nb\right)  ^{2-\varepsilon}=0$, for every
$\varepsilon>0$. It is readily to check that bounded relative error implies
asymptotically optimality.


\subsection{Two-step importance sampling for tail probabilities}

\subsubsection{First step: twisting $V$}

As a first step in providing our IS algorithm for LT-Archimedean copulas, we
apply IS to the distribution of random variable $V$. In Assumption \ref{A2},
we assume the generator $\phi(1-\frac{1}{\cdot})\in \mathrm{RV}_{-\alpha}$,
where $\phi^{-1}$ is the LS transform of random variable $V$. Then by
Karamata's Tauberian Theorem  \cite[see][pp.442--446]{feller1971introduction} $V$ is actually heavy-tailed with tail index $1/\alpha$. As noted
in \cite{asmussen2000rare}, traditional exponential twisting approach cannot
work directly for distributions with heavy tails, since a finite cumulant
generating function in \eqref{expontial change of measure} does not exist when
a positive twisting parameter is required. So an alternative method must be
used. In this subsection we describe an IS algorithm to assign a larger
probability to the event $\left \{  V>\frac{v^{\ast}}{\phi(1-f_{n})}\right \}  $
by hazard rate twisting the original distribution of $V$; see
\cite{juneja2002simulating} for an introduction on hazard rate twisting. We
prove that this leads to an estimator that is asymptotically optimal.

Let us define the hazard rate function associated to the random variable $V$
as
\[
\mathcal{H}(x)=-\log(\overline{F}_{V}(x)).
\]
By changing $\mathcal{H}(x)$ to $(1-\theta)\mathcal{H}(x)$ for some
$0<\theta<1$, the tail distribution changes to
\begin{equation}
\overline{F}_{V,\theta}(x)=(\overline{F}_{V}(x))^{1-\theta}=\exp
((\theta-1)\mathcal{H}(x)), \label{hztail}%
\end{equation}
and the density function becomes
\begin{equation}
f_{V,\theta}(x)=(1-\theta)(\overline{F}_{V}(x))^{-\theta}f_{V}(x)=(1-\theta
)\exp(\theta \mathcal{H}(x))f_{V}(x). \label{hzpdf}%
\end{equation}
Note that we have imposed an additional subscript $\theta$ on both $\overline{F}_{V,\theta}(x)$ and $\overline{f}_{V,\theta}(x)$ to emphasize that these are the functions that correspond to the transformed variable $(1-\theta)\mathcal{H}(x)$.
The prescribed transformation is similar to exponential twisting, except that the twisting rate is
$\theta \mathcal{H}(x)$ rather than $\theta x$. By \eqref{hztail}, one can also
note that the tail of random variable $V$ becomes heavier after twisting.

The key, then, is finding the best parameter $\theta$. By \eqref{hzpdf}, the
corresponding likelihood ratio $f_{V}(x)/f_{V,\theta}(x)$ is $\frac
{1}{1-\theta}\exp(-\theta \mathcal{H}(x))$, and this is upper bounded by
\begin{equation}
\frac{1}{1-\theta}\exp \left(  -\theta \mathcal{H}\left(  \frac{v^{\ast}}%
{\phi(1-f_{n})}\right)  \right)  \label{ublikelihood}%
\end{equation}
on the set $\left \{  V>\frac{v^{\ast}}{\phi(1-f_{n})}\right \}  $. It is a
standard practice in IS to search for $\tilde{\theta}$ by minimizing the upper
bound on the likelihood ratio, since this also minimizes the upper bound of
the second moment of the estimator $1_{\{L_{n}>nb\}}\frac{f_{V}(V)}%
{f_{V,\theta}^{\ast}(V)}$. By taking the derivative on the upper bound
\eqref{ublikelihood} w.r.t. $\theta$, we obtain
\[
\tilde{\theta}=1-\frac{1}{\mathcal{H}\left(  \frac{v^{\ast}}{\phi(1-f_{n}%
)}\right)  }.
\]
Then, the tail distribution in \eqref{hztail} corresponding to hazard rate
twisting by $\tilde{\theta}$ equals
\begin{equation}
\overline{F}_{V,\tilde{\theta}}(x)=\exp \left(  -\frac{\mathcal{H}%
(x)}{\mathcal{H}\left(  \frac{v^{\ast}}{\phi(1-f_{n})}\right)  }\right)  .
\label{hztheta}%
\end{equation}
Explicit form of \eqref{hztheta} is usually difficult to derive, because the
tail distribution for random variable $V$ is only specified in a
semiparametric way. Alternatively, we can replace the hazard function
$\mathcal{H}(x)$ by $\tilde{\mathcal{H}}(x)$ where $\mathcal{H}(x)\sim
\tilde{\mathcal{H}}(x)$ and $\tilde{\mathcal{H}}(x)$ is available in a closed
form. \cite{juneja2007asymptotics} prove that estimators derived by such
\textquotedblleft asymptotic\textquotedblright \ hazard rate twisting method
can achieve asymptotic optimality.

By Proposition B.1.9(1) of \cite{de2007extreme}, $\overline{F}_{V}%
\in \mathrm{RV}_{-1/\alpha}$ implies $\mathcal{H}(x)\sim \frac{1}{\alpha}%
\log(x)$ as $x\rightarrow \infty$. This, along with \eqref{hztheta}, suggests
that the tail distribution $\overline{F}_{V,\tilde{\theta}}$ should be close
to
\[
\overline{F}_{V,\tilde{\theta}}(x)\approx x^{-1/\left(  \log v^{\ast}-\log
\phi(1-f_{n})\right)  }.
\]
For considerably large $n$, we can even ignore the term $\log(v^{\ast})$ to
achieve further simplification. Hence, the corresponding density function can
be taken as
\[
\frac{1}{-\log \phi(1-f_{n})}x^{\frac{1}{\log \phi(1-f_{n})}-1},
\]
which is a Pareto distribution with shape parameter $-1/\log
\phi(1-f_{n})$. Now we define
\begin{equation}
f_{V}^{\ast}(x)=\left \{
\begin{array}
[c]{lc}%
f_{V}(x), & x<x_{0},\\
\overline{F}_{V}(x_{0})x_{0}^{-1/\log \phi(1-f_{n})}\frac{1}{-\log \phi
(1-f_{n})}x^{\frac{1}{\log \phi(1-f_{n})}-1} & x\geq x_{0},
\end{array}
\right.  \label{fvstar}%
\end{equation}
where $x_{0}$ is chosen to remain the ratio $f_{V}(x)/f_{V}^{\ast}(x)$ upper
bounded by a constant for all $x$. Thus, the tail part of random variable $V$
becomes heavier from twisting, but the probability for small values
remains unchanged.

\begin{remark}
The role of $x_{0}$ is crucial for showing the asymptotic optimality of the
algorithm, which is later seen in the proof of Lemma \ref{thasyoptimal}.
Theoretically, its value relies on the explicit expression of the density
function $f_{V}(x)$. Practically, our numerical results are not sensitive to  $x_{0}$ and hence for ease of implementation, one may fix $x_{0}$ to an
arbitrary constant.
\end{remark}

\subsubsection{Second step: twisting to Bernoulli random variables}

We now proceed to applying exponential twisting to Bernoulli random variables
$\{1_{\{U_{i}>1-l_{i}f_{n}\}},i=1,\ldots,n\}$ conditional on the common factor
$V$. A measure $\tilde{\mathbb{P}}$ is said to be an exponentially twisted
measure of $\mathbb{P}$ by parameter $\theta$, for some random variable $X$,
if
\begin{equation}
\dfrac{d\tilde{\mathbb{P}}}{d\mathbb{P}}=\exp(\theta X-\Lambda_{X}(\theta)),
\label{expontial change of measure}%
\end{equation}
where $\Lambda_{X}(\theta)=\log \mathbb{E}[\exp(\theta X)]$ represents the
cumulant generating function. Suppose random variable $X$ has density function
$f_{X}(x)$, then the exponential twisted density has the form $\exp(\theta
x-\Lambda_{X}(\theta))f_{X}(x)$.

Now we deal with the Bernoulli success probability $p_{j}$, which is
essentially $p(v,j)$ as defined in \eqref{p1} by conditioning on $V=\dfrac
{v}{\phi(1-f_{n})}$. In order to increase the conditional default
probabilities, followed by the idea in \cite{glasserman2005importance}, we
apply an exponential twist by choosing a parameter $\theta$ and taking
\[
p_{j}^{\theta}=\frac{p_{j}e^{\theta c_{j}}}{1+p_{j}\left(  e^{\theta c_{j}%
}-1\right)  },
\]
where $p_{j}^{\theta}$ denotes the $\theta$-twisted probability conditional on
$V=\dfrac{v}{\phi(1-f_{n})}$. Note that $p_{j}^{\theta}$ is a strictly
increasing function in $\theta$ if $\theta>0$. With this new choice of
conditional default probabilities $\left \{  p_{j}^{\theta}:j\leq
|\mathcal{W}|\right \}  $, straightforward calculation shows that the likelihood ratio
conditioning on $V$ simplifies to
\begin{equation}
\prod_{j\leq|\mathcal{W}|}\left(  \frac{p_{j}}{p_{j}^{\theta}}\right)
^{n_{j}Y_{j}}\left(  \frac{1-p_{j}}{1-p_{j}^{\theta}}\right)  ^{n_{j}%
(1-Y_{j})}=\exp \left(  -\theta L_{n}|V+\Lambda_{L_{n}|V}(\theta)\right)  ,
\label{exptwistln}%
\end{equation}
where
\[
\Lambda_{L_{n}|V}(\theta)=\log \mathbb{E}\left[  e^{\theta L_{n}}\left \vert
V=\frac{v}{\phi(1-f_{n})}\right.  \right]  =\sum_{j\leq|\mathcal{W}|}n_{j}%
\log \left(  1+p_{j}\left(  e^{\theta c_{j}}-1\right)  \right)
\]
is the cumulant generating function of $L_{n}$ conditional on $V$. For any
$\theta$, the estimator
\[
1_{\{L_{n}>nb|V\}}e^{-\theta L_{n}|V+\Lambda_{L_{n}|V}(\theta)}%
\]
is unbiased for $\mathbb{P}\left(  L_{n}>nb\left \vert V=\frac{v}{\phi
(1-f_{n})}\right.  \right)  $ if probabilities $\left \{  p_{j}^{\theta}%
:j\leq|\mathcal{W}|\right \}  $ are used to generate $L_{n}$. Equation
\eqref{exptwistln} formally establishes that applying an exponential twist on the
probabilities is equivalent to applying an exponential twist to $L_{n}|V$ itself.

It remains to choose the parameter $\theta$. A standard practice in IS is to
select a parameter $\theta$ that minimizes the upper bound of the second
moment of the estimator to reduce the variance. As we can see,
\[
\mathbb{E}_{\theta}\left[  1_{\{L_{n}>nb\}}e^{-2\theta L_{n}+2\Lambda_{L_{n}%
}(\theta)}\left \vert V=\frac{v}{\phi(1-f_{n})}\right.  \right]  \leq
e^{-2nb\theta+2\Lambda_{L_{n}|V}(\theta)},
\]
where $\mathbb{E}_{\theta}$ denotes expectation using the $\theta$-twisted
probabilities. The problem is then identical to finding a parameter $\theta$ that
maximizes $nb\theta-\Lambda_{L_{n}|V}(\theta)$. Straightforward calculation
shows that
\begin{equation}
\Lambda_{L_{n}|V}^{\prime}(\theta)=\sum_{j\leq|\mathcal{W}|}n_{j}c_{j}%
p_{j}^{\theta}=\mathbb{E}_{\theta}\left[  L_{n}\left \vert V=\frac{v}%
{\phi(1-f_{n})}\right.  \right]  . \label{cumulantderivative}%
\end{equation}
By the strictly increasing property of $\Lambda_{L_{n}|V}^{\prime}(\theta)$,
the maximum is attained at
\begin{equation}
\theta^{\ast}=\left \{
\begin{array}
[c]{lc}%
\text{unique solution to }\Lambda_{L_{n}|V}^{\prime}(\theta)=nb, &
nb>\Lambda_{L_{n}|V}^{\prime}(0),\\
0, & nb\leq \Lambda_{L_{n}|V}^{\prime}(0).
\end{array}
\right.  \label{thetatwocase}%
\end{equation}
By \eqref{cumulantderivative}, the two cases in \eqref{thetatwocase} are
distinguished by the value of $\mathbb{E}\left[  L_{n}\left \vert V=\frac
{v}{\phi(1-f_{n})}\right.  \right]  =\sum_{j\leq|\mathcal{W}|}n_{j}c_{j}p_{j}%
$. For the former case, our choice of twisting parameter $\theta^{\ast}$ shifts
the distribution of $L_{n}$ so that the average portfolio loss is $b$; while
for the latter case, the event $\{L_{n}>nb\}$ is not rare, so we use the
original probabilities.

\subsubsection{Algorithm\label{sssalgo}}

Now we are ready to present the algorithm. It consists of three stages. First,
a sample of $V$ is generated using hazard rate twisting. Depending on the
value of $V$, samples of the Bernoulli variables $1_{\{U_{i}>1-l_{i}f_{n}\}}$
are generated in the second step, using either naive simulation (original
probabilities) or importance sampling. The details on how to adjust
conditional default probabilities have already been discussed in the previous
subsections. Finally we compute the portfolio loss $L_{n}$ and return the
estimator after incorporating the likelihood ratio.

The following algorithm is for each replication.

\begin{center}
\fbox{\parbox{\textwidth}
{\textbf{IS Algorithm}
\begin{enumerate}
[\textit{Step} 1.]
\item Generate a sample of $V$ using the density $f_{V}^{*}$.
\item If the average portfolio loss under $V$ is greater than $b$, i.e.\[
\mathbb{E}\left[  L_{n}\left \vert V=\frac{v}{\phi(1-f_{n})}\right.  \right]
=\sum_{j\leq|\mathcal{W}|}n_{j}c_{j}p_{j}>nb,
\]
for each $i=1,2,...,n$, generate samples of $Y_{i}=1_{\{U_{i}>1-l_{i}f_{n}\}}$
independent of each other using the unchanged probability $p_{i}^{\ast}=p_{i}$. Otherwise, use $p_{i}^{\ast}=p_{i}^{\theta^{\ast}}$.
\item Calculate the portfolio loss $L_{n}=\sum_{i=1}^{n}c_{i}Y_{i}$ and return
the estimator
\begin{equation}
1_{\{L_{n}>nb\}}\frac{f_{V}(V)}{f_{V}^{\ast}(V)}\prod_{j\leq|\mathcal{W}|}\left(  \frac{p_{j}}{p_{j}^{\ast}}\right)  ^{n_{j}Y_{j}}\left(
\frac{1-p_{j}}{1-p_{j}^{\ast}}\right)  ^{n_{j}(1-Y_{j})}, \label{isestimator}\end{equation}
where $n_{j}Y_{j}$ denotes the number of defaults in sub-portfolio $j$ within
 a single simulation run.
\end{enumerate}
}}
\end{center}

\bigskip

Let $\mathbb{P}^{\ast}$ and $\mathbb{E}^{\ast}$ denote the IS\ probability
measure and expectation corresponding to this algorithm. The likelihood ratio
is given by
\[
L^{\ast}=\frac{f_{V}(V)}{f_{V}^{\ast}(V)}\prod_{j\leq|\mathcal{W}|}\left(
\frac{p_{j}}{p_{j}^{\ast}}\right)  ^{n_{j}Y_{j}}\left(  \frac{1-p_{j}}%
{1-p_{j}^{\ast}}\right)  ^{n_{j}(1-Y_{j})}.
\]
The following lemma is important in demonstrating the efficiency of our IS Algorithm.

\begin{lemma}
\label{thasyoptimal} Under the same assumptions as in Theorem \ref{th4.1}, we
have
\[
\frac{\log \mathbb{E}^{*}\left[  1_{\{L_{n}>nb\}}L^{*^{2}}\right]  }{\log
f_{n}}\to2,\quad \text{as }n\to \infty.
\]

\end{lemma}

In view of Theorem \ref{th4.1}, which provides the asymptotic estimate of the
tail probability $\mathbb{P}\left(  L_{n}>nb\right)  $, we conclude in the
following theorem that our proposed algorithm is asymptotically optimal.

\begin{theorem}
\label{coasyoptimal} Under the same assumptions as in Theorem \ref{th4.1}, we
have
\[
\lim_{n\rightarrow \infty}\frac{\log \mathbb{E}^{\ast}\left[  1_{\{L_{n}%
>nb\}}L^{\ast^{2}}\right]  }{\log \mathbb{P}\left(  L_{n}>nb\right)  }=2.
\]

\end{theorem}

Thus, the IS estimator  \eqref{isestimator} achieves asymptotic zero
variance on the logarithmic scale.

\subsection{Importance sampling for expected shortfall}

In risk management, one is usually interested in estimating the expected
shortfall at a confidence level close to $1$, which is again a rare event
simulation. In this subsection, we discuss how to apply our proposed IS
algorithm to estimate the expected shortfall.

First, note that the expected shortfall can be understood as follows,
\begin{equation}
\mathbb{E}\left[  L_{n}|L_{n}>nb\right]  =nb+\frac{\mathbb{E}\left[  \left(
L_{n}-nb\right)  _{+}\right]  }{\mathbb{P}\left(  L_{n}>nb\right)  }.
\label{es}%
\end{equation}
By involving the unbiasing likelihood ratio $L^{\ast}$, \eqref{es} is
equivalent to
\[
nb+\frac{\mathbb{E}^{\ast}\left[  \left(  L_{n}-nb\right)  _{+}L^{\ast
}\right]  }{\mathbb{E}^{\ast}\left[  1_{\{L_{n}>nb\}}L^{\ast}\right]  },
\]
where $\mathbb{E}^{\ast}$ is the expectation corresponding to the IS algorithm
in Section \ref{sssalgo}. Suppose $m$ i.i.d. samples $(L_{n}^{1},\ldots
,L_{n}^{m})$ are generated under measure $\mathbb{P}^{\ast}$. Let $L_{i}%
^{\ast}$ denote the corresponding likelihood ratio for each sample $i$. Then
the IS estimator of the expected shortfall is given as
\begin{equation}
nb+\frac{\sum_{i=1}^{m}(L_{n}^{i}-nb)_{+}L_{i}^{\ast}}{\sum_{i=1}%
^{m}1_{\{L_{n}^{i}>nb\}}L_{i}^{\ast}}. \label{ises}%
\end{equation}
Note that the samples generated to estimate the numerator in \eqref{ises} take
positive value only when large losses occur. Therefore, one can expect the IS
algorithm that works for estimating the probability of the event
$\{L_{n}>nb\}$ should also work well in estimating $\mathbb{E}[L_{n}-nb]_{+}$.
This is later confirmed by our numerical results.

\section{Conditional Monte Carlo Simulations for Large Portfolio
Loss\label{scmcs}}

In this section, we propose another estimation method based on the conditional
Monte Carlo approach, which is another variance reduction technique; see,
e.g., \cite{asmussen2006improved} and \cite{asmussen2018conditional}. Our
proposed algorithm is motivated by \cite{chan2010efficient}, in which the
authors derived simple simulation algorithms to estimate the probability of
large portfolio losses under the $t$-copula.

By utilizing the stochastic representation \eqref{representlt} for
LT-Archimedean and the asymptotic expansions in Theorem \ref{th4.1}, the rare
event $\{L_{n}>nb\}$ occurs primarily when the random variable $V$ takes
large value, while $\mathbf{R}=(R_{1},\ldots,R_{n})$ generally has little
influence on the occurrence of the rare event. This simply suggests that an
approach by integrating out $V$ analytically could lead to a substantial
variance reduction.

To proceed, it is useful to define
\begin{equation}
O_{i}=\frac{R_{i}}{\phi(1-l_{i}f_{n})},i=1,\ldots,n. \label{rtoo}%
\end{equation}
The individual obligor defaults if $U_{i}>1-l_{i}f_{n}$, then $V>O_{i}$. Thus,
the portfolio loss in \eqref{m1} can be rewritten as,
\[
L_{n}=\sum_{i=1}^{n}c_{i}1_{\{V>O_{i}\}}.
\]
We rank $O_{1},\ldots,O_{n}$ as $O_{(1)}\leq O_{(2)}\leq \cdots \leq O_{(n)}$,
and let $c_{(i)}$ denote the associated exposure at default with $O_{(i)}$.
Then, one can check that the event $\{L_{n}>nb\}$ happens if and only if
$V>O_{(k)}$, where $k=\min \{l:\sum_{i=1}^{l}c_{(i)}>nb\}$. Particularly, if
$c_{i}\equiv c$ for all $i=1,\ldots,n$, then $k=\lceil nb/c\rceil$. Now
conditional on $\mathbf{R}$, we have
\begin{equation}
\mathbb{P}\left(  L_{n}>nb|\mathbf{R}\right)  =\mathbb{P}\left(
V>O_{(k)}|\mathbf{R}\right)  :=S(\mathbf{R}). \label{condmc1est}%
\end{equation}
We summarize our proposed conditional Monte Carlo algorithm, which is labelled as CondMC, in the following algorithm.

\begin{center}
\fbox{\parbox{\textwidth}{\textbf{Conditional Monte Carlo (CondMC) Algorithm}
\begin{enumerate}[\textit{Step} 1.]
\item Generate independent standard exponential random variables $R_1,\ldots,R_n$.
\item For $i=1,\ldots,n$, transform $R_i$ to $O_i$ according to \eqref{rtoo}.
\item Find $O_{(k)}$ and return the conditional Monte Carlo estimator $S(\mathbf{R})$ based on \eqref{condmc1est}.
\end{enumerate}}}
\end{center}

We now show that the conditional Monte Carlo estimator has bounded relative
error, a stronger notion of asymptotic optimality than that for the IS
estimator (\ref{isestimator}) established in Theorem \ref{coasyoptimal}.

\begin{lemma}
\label{thboundederror} Under the same assumptions as in Theorem \ref{th4.1}
except that $\frac{1}{n}=O(f_{n})$, we have
\[
\limsup_{n\to \infty}\frac{\mathbb{E}\left[  S^{2}(\mathbf{R})\right]  }%
{f_{n}^{2}}<\infty.
\]

\end{lemma}

In view of Theorem \ref{th4.1}, we immediately obtain the following theorem
concerning the algorithm efficiency.

\begin{theorem}
\label{coboundederror} Under the same assumptions as in Lemma
\ref{thboundederror}, we have
\[
\limsup_{n\rightarrow \infty}\frac{\sqrt{\mathbb{E}\left[  S^{2}(\mathbf{R}%
)\right]  }}{\mathbb{P}\left(  L_{n}>nb\right)  }<\infty.
\]
In other words, the conditional Monte Carlo estimator \eqref{condmc1est} has
bounded relative error.
\end{theorem}

\section{Numerical Results\label{snr}}

{In this section, we assess the relative performance of our proposed algorithms via simulations, and investigate their sensitivity to $\alpha$ (heavy tailedness
of the systematic risk factor $V$), $n$ (size of the portfolio) and $b$ (a
pre-fixed number that controls the level of the proportion of obligors who
default). The numerical results indicate that our proposed algorithms, especially the
CondMC algorithm, provide considerable variance reductions when compared to
naive MC simulations. This supports our theoretical result that our proposed
algorithms are all asymptotically optimal.}


Due to the assumption that $\phi(1-\frac{1}{\cdot})\in \mathrm{RV}_{-\alpha}$
with $\alpha>1$, we consider the Gumbel copula in our numerical experiment.
The generator function of Gumbel copula is $\phi(t)=(-\ln(t))^{\alpha}$ with
$\alpha>1$. By varying $\alpha$, the Gumbel copula covers from independence
($\alpha \rightarrow1$) to comonotonicity ($\alpha \rightarrow \infty$).

In all the experiments below, only homogeneous portfolios are considered.
However, it should be emphasized that the performance of our algorithms is not
affected for inhomogeneous portfolio. This is asserted by Theorem
\ref{coasyoptimal} and Theorem \ref{coboundederror} that have been proved under a general
setting for both homogeneous and inhomogeneous portfolios. To evaluate the
accuracy of the estimators, for each set of specified parameters, we generate
50,000 samples for our proposed algorithms, estimate the probability of large
portfolio loss, and provide the relative error (in $\%$), which is defined as
the ratio of the estimator's standard deviation to the estimator. More
precisely, if $\hat{p}$ is an unbiased estimator of $\mathbb{P}\left(
L_{n}>nb\right)  $, its relative error is defied as $\sqrt{\mathrm{Var}%
(\hat{p})}/\hat{p}$. We also report the variance reduction achieved by our
proposed algorithms compared with naive simulation. For naive simulation, it
is highly possible that the rare event would not be observed in any sample
path with only 50,000 samples. Therefore, variance under naive simulation is
estimated indirectly by exploiting the fact that variance for Bernoulli($p$)
equals $p(1-p)$.

Table \ref{t2} provides a first comparison of our IS algorithm and CondMC algorithm
with naive simulation as $\alpha$ changes. The chosen model parameter values 
are $n=500$, $f_{n}=1/n$, $b=0.8$, $l_{i}=0.5$ and $c_{i}=1$ for each $i$. As
can be concluded from Table \ref{t2}, both algorithms outperform the naive
simulation, especially when $\alpha$ is small, obligors have weaker dependence
and the probability of large portfolio losses becomes smaller. Relative to the naive MC method,  the variance reduction attained by the IS estimator is in the order of hundreds and thousands while the CondMC estimator is in the order of millions. 
This demonstrates that CondMC estimator significantly outperforms IS estimator. 

\begin{table}[ptbh]
\begin{center}%
\begin{tabular}
[c]{lccccrr}\hline
& \multicolumn{2}{c}{Prob. estimate} & \multicolumn{2}{c}{Relative error (\%)}
& \multicolumn{2}{c}{Variance reduction}\\ \cline{2-7}%
$\alpha$ & IS & CondMC & IS & CondMC & IS & CondMC\\ \hline
1.1 & 6.112$\times10^{-5}$ & 6.208$\times10^{-5}$ & 1.468 & 0.023 & 1,519 &
6,248,304\\
1.5 & 2.652$\times10^{-4}$ & 2.726$\times10^{-4}$ & 1.554 & 0.017 & 312 &
2,658,936\\
2 & 4.436$\times10^{-4}$ & 4.457$\times10^{-4}$ & 1.542 & 0.012 & 189 &
2,910,515\\
5 & 7.706$\times10^{-4}$ & 7.815$\times10^{-4}$ & 1.575 & 0.005 & 105 &
10,338,790\\ \hline
\end{tabular}
\end{center}
\caption{Performance of the proposed algorithms for Gumbel copula under
different values of $\alpha$.}%
\label{t2}%
\end{table}

In Table \ref{t3}, we perform the same comparison by varying $b$ while keeping
$\alpha$ fixed at $1.5$. Under the setting that $c=1$, the parameter $b$
controls the level of the proportion of obligors that default. As is clear from
the table, when $b$ increases, the estimated probability decreases and the
variance reduction becomes larger. \begin{table}[ptbh]
\begin{center}%
\begin{tabular}
[c]{ccccccr}\hline
& \multicolumn{2}{c}{Prob. estimate} & \multicolumn{2}{c}{Relative error (\%)}
& \multicolumn{2}{c}{Variance reduction}\\ \cline{2-7}%
$b$ & IS & CondMC & IS & CondMC & IS & CondMC\\ \hline
0.3 & 7.415$\times10^{-4}$ & 7.437$\times10^{-4}$ & 1.414 & 0.024 & 135 &
447,754\\
0.5 & 4.714$\times10^{-4}$ & 4.776$\times10^{-4}$ & 1.462 & 0.019 & 198 &
1,130,242\\
0.7 & 3.293$\times10^{-4}$ & 3.306$\times10^{-4}$ & 1.506 & 0.017 & 268 &
2,129,103\\
0.9 & 2.101$\times10^{-4}$ & 2.151$\times10^{-4}$ & 1.569 & 0.017 & 386 &
3,090,169\\ \hline
\end{tabular}
\end{center}
\caption{Performance of the proposed algorithms for Gumbel copula under
different values of $b$.}%
\label{t3}%
\end{table}

Table \ref{t4} provides the relative error and variance reduction of our
algorithms compared with naive simulation as the number of obligors changes.
All other parameters are identical to previous experiments by fixing
$\alpha=1.5$ and $b=0.8$. In the last column, we also derive the sharp
asymptotic for the desired probability of large portfolio loss based on the
expression in \eqref{sa1}. Note that as $n$ increases, both the accuracy of
the sharp asymptotic and the reduction in variance improve. \begin{table}[ptbh]
\begin{center}%
\begin{tabular}
[c]{llccccrrl}\hline
& \multicolumn{2}{c}{Prob. estimate} & \multicolumn{2}{c}{Relative error (\%)}
& \multicolumn{2}{c}{Variance reduction} & \\ \cline{2-7}%
$n$ & IS & CondMC & IS & CondMC & IS & CondMC & Asymptotic\\ \hline
100 & 1.373$\times10^{-3}$ & 1.381$\times10^{-3}$ & 1.398 & 0.037 & 74 &
105,710 & 1.359$\times10^{-3}$\\
250 & 5.372$\times10^{-4}$ & 5.470$\times10^{-4}$ & 1.487 & 0.023 & 168 &
670,052 & 5.436$\times10^{-4}$\\
500 & 2.723$\times10^{-4}$ & 2.727$\times10^{-4}$ & 1.529 & 0.017 & 314 &
2,671,423 & 2.718$\times10^{-4}$\\
1,000 & 1.356$\times10^{-4}$ & 1.361$\times10^{-4}$ & 1.640 & 0.012 & 582 &
10,608,750 & 1.359$\times10^{-4}$\\ \hline
\end{tabular}
\end{center}
\caption{Performance of the proposed algorithms for Gumbel copula together
with the sharp asymptotic derived in Theorem \ref{th4.1} under different
values of $n$.}%
\label{t4}%
\end{table}

In Table \ref{t5}, we study the accuracy of the sharp asymptotic for expected
shortfall as the number of obligors increases. Model parameters are taken to
be $f_{n}=1/n$, $\alpha=1.5$, $b=0.8$, $l_{i}=0.5$ and $c=1$ for each $i$. For
estimating expected shortfall, we simply use all the 50,000 sample paths
generated under the proposed IS measure, and then consider those with
portfolio loss exceeding $nb$. As shown in Table \ref{t5}, the accuracy is
quite high even for small values of $n$. This is mainly due to the fact that
the hazard rate density is chosen based on the asymptotic result in Theorem
\ref{th4.1}. Discrepancy here is measured as the percentage difference between
the ES estimated via importance sampling and the sharp asymptotic in \eqref{essa}.

\begin{table}[ptbh]
\begin{center}%
\begin{tabular}
[c]{rrrc}\hline
$n$ & ES estimate & Asymptotic & Discrepancy(\%) \\ \hline
50 & 47.886 & 47.695 & 0.399  \\
100 & 95.573 & 95.390 & 0.191  \\
250 & 238.873 & 238.475 & 0.167  \\
500 & 477.558 & 476.950 & 0.127  \\ \hline
\end{tabular}
\end{center}
\caption{The expected shortfall and its sharp asymptotic derived in Theorem
\ref{th4.2} under different values of $n$.}%
\label{t5}%
\end{table}

{To conclude the section, we note again that to the best of our knowledge, this is the first paper that adopts the Archimedean copula in the analysis of the large credit portfolio loss and proposes the corresponding importance sampling and conditional Monte Carlo estimators. On the other hand, the importance sampling estimators of \cite{bassamboo2008portfolio}  and conditional Monte Carlo estimators of \cite{chan2010efficient} assume that the dependence structure of obligors are modeled with a $t$-copula. Because of the difference in the underlying assumed dependence structure, 
the estimators considered in this paper are not directly comparable to the corresponding estimators in those two papers.  Nevertheless, by comparing our simulation results to theirs, it is reassuring that even under very different dependence structure, significant variance reduction, especially for the estimator based on the conditional Monte Carlo method, 
can be expected. Furthermore, regardless of the assumed dependence structure, 
all estimators exhibit consistent behavior in the sense that they perform better for weaker dependence structures and larger portfolio sizes. 
}

\section{Conclusion}

In this paper, we consider an Archimedean copula-based model for measuring
portfolio credit risk. The analytic expressions of the probability of such
portfolio incurs large losses is not available and directly applying naive MC
simulation on these rare events are also not efficient. We first derive sharp
asymptotic expansions to study the probability of large portfolio losses and
the expected shortfall of the losses. Using this as a stepping stone, we
develop two efficient algorithms to estimate the risk of a credit portfolio
via simulation. The first one is a two-step full IS algorithm, which can be
used to estimate both  probability and expected shortfall of
portfolio loss. We show that the proposed estimator is logarithmically
efficient. The second algorithm is based on the conditional Monte Carlo
simulation, which can be used to estimate the probability of portfolio loss.
This estimator is shown to have bounded relative error. Through extensive
simulation studies, both algorithms, especially the second one, show
significant variance reductions when compared to naive MC simulations.

\appendix
\renewcommand{\thesubsection}{A.\arabic{subsection}}
\renewcommand{\theequation}{A.\arabic{equation}} \renewcommand{\thelemma}{A.\arabic{lemma}}

\section*{Appendix: Proofs}

To simplify the notation, for any two positive functions $g$ and $h$, we write
$g\lesssim h$ or $h\gtrsim g$ if $\lim \sup g/h\leq1$. \bigskip

\subsection{Proofs for LT-Archimedean copulas}

We first list a series of lemmas that will be useful for proving Theorem \ref{th4.1} and Theorem
\ref{th4.2}. The following is a restatement of Theorem 2 of
\cite{hoeffding1963probability}.

\begin{lemma}
If $X_{1},X_{2},\ldots,X_{n}$ are independent and $a_{i}\leq X_{i}\leq b_{i}$
for $i=1,\ldots,n$, then for $\varepsilon>0$
\[
\mathbb{P}\left(  \left \vert \bar{X}_{n}-\mathbb{E}\left[  \bar{X}_{n}\right]
\right \vert \geq \varepsilon \right)  \leq2\exp \left(  -\frac{2n^{2}%
\varepsilon^{2}}{\sum_{i=1}^{n}(b_{i}-a_{i})^{2}}\right)  ,
\]
\label{leho} with $\bar{X}_{n}=\left(  X_{1}+X_{2}+\ldots+X_{n}\right)  /n$.
\end{lemma}

Applying Lemma \ref{leho}, we obtain the following inequality:

\begin{lemma}
\label{leld} For any $\varepsilon>0$ and any large $M$, there exists a
constant $\beta>0$ such that
\[
\mathbb{P}_{v}\left(  \left \vert \frac{1}{n}\sum_{i=1}^{n} c_{i}
1_{\{U_{i}>1-l_{i}f_{n}\}}-r(v)\right \vert \geq \varepsilon \right)  \leq
\exp(-n\beta),
\]
uniformly for all $0<v\leq M$ and for all sufficiently large $n$, where
$\mathbb{P}_{v}$ denotes the original probability measure conditioned on
$V=\frac{v}{\phi(1-f_{n})}$.
\end{lemma}

\begin{proof}
Note that $U_{i}$ are conditionally independent on $V$. Then by Lemma
\ref{leho}, for every $n$,
\begin{equation}
\mathbb{P}_{v}\left(  \left \vert \frac{1}{n}\sum_{i=1}^{n} c_{i}
1_{\{U_{i}>1-l_{i}f_{n}\}}-\frac{1}{n}\sum_{i=1}^{n}c_{i}p(v,i)\right \vert
\geq2\varepsilon \right)  \leq2\exp \left(  -\frac{8n^{2}\varepsilon^{2}}%
{\sum_{i=1}^{n} c_{i}^{2}}\right)  \leq \exp(-n\beta), \label{lea2.1}%
\end{equation}
where $\beta$ is some unimportant constant not depending on $n$ and $v$.

Using \eqref{lea2.1}, to obtain the desired result, it suffices to show the
existence of $N$, such for all $n\geq N$,
\begin{equation}
\left \vert \frac{1}{n}\sum_{i=1}^{n}c_{i}p(v,i)-r(v)\right \vert \leq
\varepsilon \label{lea2.2}%
\end{equation}
holds uniformly for all $v\leq M$. Recall that $r(v)=\sum_{j\leq|\mathcal{W}%
|}c_{j}w_{j}\tilde{p}(v,j)$. Note that $n_{j}$ denotes the number of obligors
in sub-portfolio $j$. Then
\begin{align}
\left \vert \frac{1}{n}\sum_{i=1}^{n}c_{i}p(v,i)-r(v)\right \vert  &
=\left \vert \sum_{j\leq|\mathcal{W}|}c_{j}\left(  p(v,j)\frac{n_{j}}{n}%
-\tilde{p}(v,j)w_{j}\right)  \right \vert \nonumber \\
&  \leq \sum_{j\leq|\mathcal{W}|}c_{j}p(v,j)\left \vert \frac{n_{j}}{n}%
-w_{j}\right \vert \nonumber \\
&  +\sum_{j\leq|\mathcal{W}|}c_{j}w_{j}\left \vert p(v,j)-\tilde{p}%
(v,j)\right \vert \nonumber \\
&  \leq \sum_{j\leq|\mathcal{W}|}c_{j}\left \vert \frac{n_{j}}{n}-w_{j}%
\right \vert +\bar{c}\max \limits_{j\leq|\mathcal{W}|}\left \vert p(v,j)-\tilde
{p}(v,j)\right \vert \label{lea2.3}%
\end{align}
where $\bar{c}=\sum_{j\leq|\mathcal{W}|}c_{j}w_{j}$. By Assumption \ref{ass1},
there exists $N_{1}$ satisfying $\sum_{j\leq|\mathcal{W}|}c_{j}\left \vert
\frac{n_{j}}{n}-w_{j}\right \vert \leq \frac{\varepsilon}{2}$ for all $n\geq
N_{1}$. For the second part of \eqref{lea2.3}, by noting that $e^{x}\geq1+x$
for all $x\in \mathbb{R}$, we have
\begin{align*}
\left \vert p(v,j)-\tilde{p}(v,j)\right \vert  &  =\exp \left(  -v\left(
\frac{\phi(1-l_{j}f_{n})}{\phi(1-f_{n})}\wedge l_{j}^{\alpha}\right)  \right)
\left(  1-\exp \left(  -v\left \vert \frac{\phi(1-l_{j}f_{n})}{\phi(1-f_{n}%
)}-l_{j}^{\alpha}\right \vert \right)  \right) \\
&  \leq v\left \vert \frac{\phi(1-l_{j}f_{n})}{\phi(1-f_{n})}-l_{j}^{\alpha
}\right \vert \\
&  \leq M\left \vert \frac{\phi(1-l_{j}f_{n})}{\phi(1-f_{n})}-l_{j}^{\alpha
}\right \vert .
\end{align*}
Since $\phi \in \mathrm{RV}_{\alpha}(1)$, there exists $N_{2}$ such that for all
$n\geq N_{2}$, $\bar{c}\max \limits_{j\leq|\mathcal{W}|,v\in A}\left \vert
p(v,j)-\tilde{p}(v,j)\right \vert \leq \frac{\varepsilon}{2}$.

Combining the upper bound for both parts in \eqref{lea2.3} and letting
$N=\max \{N_{1},N_{2}\}$, \eqref{lea2.2} holds uniformly for all $v\leq M$. The
proof is then completed.\bigskip
\end{proof}

The following proof of Theorem \ref{th4.1} is motivated by the proof of
Theorem 1 in \cite{bassamboo2008portfolio}.

\begin{proof}
[Proof of Theorem \ref{th4.1}]Let $v_{\delta}^{\ast}$ denote the unique
solution to the equation $r(v)=b-\delta$. By using continuity and monotonicity
of $r(v)$ in $v$, we have
\[
v_{\delta}^{\ast}\rightarrow v^{\ast}%
\]
as $\delta \rightarrow0$.

Fix $\delta>0$. We decompose the probability of the event $\{L_{n}>nb\}$ into
two terms as
\begin{align*}
\mathbb{P}\left(  L_{n}>nb\right)   &  =\mathbb{P}\left(  L_{n}>nb,V\leq
\frac{v_{\delta}^{\ast}}{\phi(1-f_{n})}\right)  +\mathbb{P}\left(
L_{n}>nb,V>\frac{v_{\delta}^{\ast}}{\phi(1-f_{n})}\right) \\
&  =I_{1}+I_{2}.
\end{align*}
The remaining part of proof will be divided into three steps. We first show
that $I_{1}$ is asymptotically negligible. Then we develop upper and lower
bounds for $I_{2}$ with the second and third step.

\noindent \textbf{Step 1.} We show $I_{1}=o(f_{n})$. Note that for any $v\leq
v_{\delta}^{\ast}$, $r(v)\leq b-\delta$. Thus, by Lemma \ref{leld}, for all
sufficiently large $n$, there exists a constant $\beta>0$ such that
\[
\mathbb{P}_{v}\left(  L_{n}>nb\right)  \leq \mathbb{P}_{v}\left(  \frac{1}%
{n}\sum_{i=1}^{n}c_{i}1_{\{U_{i}>1-l_{i}f_{n}\}}-r(v)>\delta \right)  \leq
\exp(-n\beta)
\]
uniformly for all $v\leq v_{\delta}^{\ast}$. So the same upper bound holds for
$I_{1}$. Due to the condition on $f_{n}$, $I_{1}=o(f_{n})$.

\noindent \textbf{Step 2.} We now develop an asymptotic upper bound for $I_{2}%
$. Note that
\[
I_{2}\leq \mathbb{P}\left(  V>\frac{v_{\delta}^{\ast}}{\phi(1-f_{n})}\right)
=\overline{F}_{V}\left(  \frac{v_{\delta}^{\ast}}{\phi(1-f_{n})}\right)  .
\]
Recall that $\phi^{-1}$ is the LS transform for random variable $V$. Then by
$\phi(1-\frac{1}{\cdot})\in \mathrm{RV}_{-\alpha}$ and Karamata's Tauberian
theorem, we obtain
\begin{align*}
I_{2}  &  \leq \overline{F}_{V}\left(  \frac{v_{\delta}^{\ast}}{\phi(1-f_{n}%
)}\right) \\
&  \sim \frac{{1-\phi^{-1}\left(  \frac{\phi(1-f_{n})}{v_{\delta}^{\ast}%
}\right)  }}{{\Gamma(1-1/\alpha)}}\\
&  \sim f_{n}\frac{(v_{\delta}^{\ast})^{-1/\alpha}}{\Gamma(1-1/\alpha)},
\end{align*}
where in the first step we used $\overline{F}_{V}\in \mathrm{RV}_{-1/\alpha}$
and the second step is due to $1-\phi^{-1}(\frac{1}{\cdot})\in \mathrm{RV}%
_{1/\alpha}$. Letting $\delta \downarrow0$, we obtain
\begin{equation}
I_{2}\lesssim f_{n}\frac{(v^{\ast})^{-1/\alpha}}{\Gamma(1-1/\alpha)}.
\label{th3.1step2}%
\end{equation}
\noindent \textbf{Step 3.} We now develop an asymptotic lower bound for $I_{2}%
$. Denote $v_{\widehat{\delta}}^{\ast}$ as the unique solution to the equation
$r(v)=b+\delta$. Similarly, we have $v_{\widehat{\delta}}^{\ast}\rightarrow
v^{\ast}$ as $\delta \rightarrow0$. It also follows from the monotonicity of
$r(v)$ that $v_{\widehat{\delta}}^{\ast}\geq v_{\delta}^{\ast}$. Thus,
\[
I_{2}\geq \mathbb{P}\left(  L_{n}>nb,V>\frac{v_{\widehat{\delta}}^{\ast}}%
{\phi(1-f_{n})}\right)  .
\]
Note that for any large $M>0$, applying Lemma \ref{leld}, it holds uniformly
for $v\in \left[  v_{\hat{\delta}}^{\ast},M\right]  $ that $r(v)\geq b+\delta$
and then as $n\rightarrow \infty$, by Lemma \ref{leld}
\begin{align*}
\mathbb{P}_{v}\left(  L_{n}>nb\right)   &  \geq \mathbb{P}_{v}\left(  \frac
{1}{n}\sum_{i=1}^{n}c_{i}1_{\{U_{i}>1-l_{i}f_{n}\}}-r(v)>-\delta \right) \\
&  =1-\mathbb{P}_{v}\left(  \frac{1}{n}\sum_{i=1}^{n}c_{i}1_{\{U_{i}%
>1-l_{i}f_{n}\}}-r(v)\leq-\delta \right)  \rightarrow1.
\end{align*}
Hence,
\begin{align*}
I_{2}  &  \gtrsim \overline{F}_{V}\left(  \frac{v_{\hat{\delta}}^{\ast}}%
{\phi(1-f_{n})}\right)  -\overline{F}_{V}\left(  \frac{M}{\phi(1-f_{n}%
)}\right) \\
&  \sim f_{n}\frac{(v_{\hat{\delta}}^{\ast})^{-1/\alpha}}{\Gamma(1-1/\alpha
)}-f_{n}\frac{M^{-1/\alpha}}{\Gamma(1-1/\alpha)}.
\end{align*}
Taking $M\rightarrow \infty$ followed by $\delta \rightarrow0$, we get
\begin{equation}
I_{2}\gtrsim f_{n}\frac{(v^{\ast})^{-1/\alpha}}{\Gamma(1-1/\alpha)}.
\label{th3.1step3}%
\end{equation}

Combining \eqref{th3.1step2}, \eqref{th3.1step3} with Step 1 completes the
proof of the theorem.\bigskip
\end{proof}

\begin{proof}
[Proof of Theorem \ref{th4.2}]We first note that the expected shortfall can be
rewritten as in (\ref{essurvival}). Using Theorem \ref{th4.1}, in order to get
the desired result, it suffices to show that
\begin{equation}
\int_{b}^{\infty}\mathbb{P}\left(  L_{n}>nx\right)  \mathrm{d}x\sim f_{n}%
\frac{\int_{v^{\ast}}^{\infty}r^{\prime}(v)v^{-1/\alpha}\mathrm{d}v}%
{\Gamma(1-1/\alpha)}. \label{intsurvival}%
\end{equation}
We decompose the left-hand side of \eqref{intsurvival} into the following two
terms
\begin{align*}
\int_{b}^{\infty}\mathbb{P}\left(  L_{n}>nx\right)  \mathrm{d}x  &  =\int
_{b}^{\bar{c}}\mathbb{P}\left(  L_{n}>nx\right)  \mathrm{d}x+\int_{\bar{c}%
}^{\infty}\mathbb{P}\left(  L_{n}>nx\right)  \mathrm{d}x\\
&  :=J_{1}+J_{2},
\end{align*}
where $\bar{c}=\sum_{j\leq|\mathcal{W}|}c_{j}w_{j}$. The remaining part of
proof will be divided into three steps. We first show $\mathbb{P}\left(
L_{n}>n\bar{c}\right)  $ and $J_{2}$ are asymptotically negligible in the
first two steps. Then we develop the asymptotic for $J_{1}$ in the last step.
For simplicity, we denote the unique solution of the equation $r(v)=s$ for
$0\leq s\leq \bar{c}$ by $r^{\leftarrow}(s)$.

\noindent \textbf{Step 1.} In this step, we show
\begin{equation}
\mathbb{P}\left(  L_{n}>n\bar{c}\right)  =o(f_{n}). \label{asycbar}%
\end{equation}
Fix an arbitrarily small $\delta>0$. Proceeding in the same way as in step 1
in the proof of Theorem \ref{th4.1}, for all sufficiently large $n$, there
exists a constant $\beta>0$ such that
\[
\mathbb{P}\left(  L_{n}>n\bar{c},V\leq \frac{r^{\leftarrow}(\bar{c}-\delta
)}{\phi(1-f_{n})}\right)  \leq \exp(-n\beta).
\]
Due to the condition on $f_{n}$ and letting $\delta \downarrow0$, we have the
desired result in \eqref{asycbar}.

\noindent \textbf{Step 2.} In this step, we show $J_{2}=o(f_{n}).$ Note that
$J_{2}$ can be rewritten as follows,
\begin{align*}
J_{2}  &  =\mathbb{E}\left[  \left(  \frac{L_{n}}{n}-\bar{c}\right)
_{+}\right] \\
&  =\mathbb{E}\left[  \left(  \frac{L_{n}}{n}-\bar{c}\right)  1_{\left \{
L_{n}>n\bar{c}\right \}  }\right]  .
\end{align*}
Since $\frac{L_{n}}{n}<\max \limits_{j\leq \vert \mathcal{W}\vert}c_{j}$, we
have
\begin{align*}
J_{2}\leq \left(  \max \limits_{j\leq \vert \mathcal{W}\vert}c_{j}-\bar{c}\right)
\mathbb{P}\left(  L_{n}>n\bar{c}\right)  .
\end{align*}
It follows from \eqref{asycbar} that $J_{2}=o(f_{n})$.

\noindent \textbf{Step 3.} To this end, we show
\[
\lim_{n\rightarrow \infty}\int_{b}^{\bar{c}}\frac{\Gamma(1-1/\alpha)}{f_{n}%
}\mathbb{P}\left(  L_{n}>nx\right)  \mathrm{d}x=\int_{v^{\ast}}^{\infty
}r^{\prime}(v)v^{-1/\alpha}\mathrm{d}v.
\]
First note that, for any $x\in \lbrack b,\bar{c}]$, by Theorem \ref{th4.1} we
have
\[
\lim_{n\rightarrow \infty}\frac{\Gamma(1-1/\alpha)}{f_{n}}\mathbb{P}\left(
L_{n}>nx\right)  =(r^{\leftarrow}(x))^{-1/\alpha}.
\]
Further, the following inequality holds any $x\in \lbrack b,\bar{c}]$
\[
\frac{\Gamma(1-1/\alpha)}{f_{n}}\mathbb{P}\left(  L_{n}>nx\right)  \leq
\frac{\Gamma(1-1/\alpha)}{f_{n}}\mathbb{P}\left(  L_{n}>nb\right)  .
\]
Applying the dominated convergence theorem, we obtain
\begin{align*}
\lim_{n\rightarrow \infty}\int_{b}^{\bar{c}}\frac{\Gamma(1-1/\alpha)}{f_{n}%
}\mathbb{P}\left(  L_{n}>nx\right)  \mathrm{d}x  &  =\int_{b}^{\bar{c}}\left(
\lim_{n\rightarrow \infty}\frac{\Gamma(1-1/\alpha)}{f_{n}}\mathbb{P}\left(
L_{n}>nx\right)  \right)  \mathrm{d}x\\
&  =\int_{b}^{\bar{c}}(r^{\leftarrow}(x))^{-1/\alpha}\mathrm{d}x\\
&  =\int_{v^{\ast}}^{\infty}r^{\prime}(v)v^{-1/\alpha}\mathrm{d}v.
\end{align*}
The last equality is by changing the variable and let $v=r^{\leftarrow}(x)$.

Combing Step 2 and Step 3 completes the proof of the theorem.\bigskip
\end{proof}


\subsection{Proofs for algorithm efficiency}

Lemma \ref{leratiobound} and \ref{lephifn} will be used in proving Lemma
\ref{thasyoptimal}.

\begin{lemma}
\label{leratiobound}For sufficiently large $n$, there exists a constant $C$
such that
\begin{equation}
\frac{f_{V}(x)}{f_{V}^{\ast}(x)}\leq C\left(  -\log \phi(1-f_{n})\right)
\label{ratioC}%
\end{equation}
for all $x$, where $f_{V}^{\ast}(x)$ is defined in \eqref{fvstar}.
\end{lemma}

\begin{proof}
By definition of $f_{V}^{\ast}(x)$, the ratio $\frac{f_{V}(x)}{f_{V}^{\ast
}(x)}$ equals $1$ for $x<x_{0}$. Hence, to show \eqref{ratioC}, it suffices to
show the existence of a constant $C$ for all $x\geq x_{0}$.

Note that when $x\geq x_{0}$,
\[
\frac{f_{V}(x)}{f_{V}^{\ast}(x)}=\frac{f_{V}(x)}{\overline{F}_{V}(x_{0})}%
x_{0}^{1/\log \phi(1-f_{n})}\left(  -\log \phi(1-f_{n})\right)  x^{1-\frac
{1}{\log \phi(1-f_{n})}}.
\]
By Assumption \ref{A2} that $V$ has a eventually monotone density function, we
have $f_{V}\in \mathrm{RV}_{-1/\alpha-1}$. Then by Potter's bounds (see e.g.
Theorem B.1.9 (5) of \cite{de2007extreme}), for any small $\varepsilon>0$,
there exists $x_{0}>0$ and a constant $C_{0}>0$ such that for all $x\geq
x_{0}$
\[
f_{V}(x)\leq C_{0}x^{-\frac{1}{\alpha}-1+\varepsilon}.
\]
Thus,
\begin{align}
\frac{f_{V}(x)}{f_{V}^{\ast}(x)}  &  \leq \frac{C_{0}}{\overline{F}_{V}(x_{0}%
)}x_{0}^{1/\log \phi(1-f_{n})}\left(  -\log \phi(1-f_{n})\right)  x^{-1/\alpha
-\frac{1}{\log \phi(1-f_{n})}+\varepsilon}\label{ratiovdelta}\\
&  \leq C\left(  -\log \phi(1-f_{n})\right)  ,\nonumber
\end{align}
which yields our desired result by noting the fact that $x\geq x_{0}$ and
$-1/\alpha-\frac{1}{\log \phi(1-f_{n})}+\varepsilon<0$. \bigskip
\end{proof}

\begin{lemma}
\label{lephifn} If $\phi(1-\frac{1}{\cdot})\in \mathrm{RV}_{-\alpha}$ for some
$\alpha>1$ and $f_{n}$ is a positive deterministic function converging to $0$
as $n\rightarrow \infty$, then
\[
\log \phi(1-f_{n})\sim \alpha \log(f_{n}).
\]

\end{lemma}

\begin{proof}
By Proposition B.1.9(1) of \cite{de2007extreme}, $\phi \in \mathrm{RV}_{\alpha
}(1)$ implies that
\[
\log \phi(1-x)\sim \alpha \log(x)
\]
as $x\rightarrow0$. \bigskip
\end{proof}

The following proof is motivated by the proof of Theorem 3 in
\cite{bassamboo2008portfolio}.

\begin{proof}
[Proof of Lemma \ref{thasyoptimal}]Let
\[
\hat{L}=\prod_{j\leq|\mathcal{W}|}\left(  \frac{p_{j}}{p_{j}^{\ast}}\right)
^{n_{j}Y_{j}}\left(  \frac{1-p_{j}}{1-p_{j}^{\ast}}\right)  ^{n_{j}(1-Y_{j}%
)}.
\]
Note that if $\mathbb{E}\left[  L_{n}\left \vert V=\frac{v}{\phi(1-f_{n}%
)}\right.  \right]  <nb$, $p_{j}^{\ast}=p_{\theta^{\ast}}(V\phi(1-f_{n}),j)$
where $\theta^{\ast}$ is chosen by solving $\Lambda_{L_{n}|V}^{\prime}%
(\theta)=nb$; otherwise $p_{j}^{\ast}=p\left(  V\phi(1-f_{n}),j\right)  $ by
setting $\theta^{\ast}=0$. Besides, \eqref{exptwistln} shows $\hat{L}$ can be
written as follows.
\[
\hat{L}=\exp(-\theta^{\ast}L_{n}|V+\Lambda_{L_{n}|V}(\theta^{\ast})).
\]
Then it follows that, for any $v$,
\[
1_{\left \{  L_{n}>nb,V=\frac{v}{\phi(1-f_{n})}\right \}  }\hat{L}%
\leq1_{\left \{  L_{n}>nb,V=\frac{v}{\phi(1-f_{n})}\right \}  }\exp
(-\theta^{\ast}nb+\Lambda_{L_{n}|V}(\theta^{\ast}))\qquad \text{a.s.}%
\]
Since $\Lambda_{L_{n}|V}(\theta)$ is a strictly convex function, one can
observe that $-\theta nb+\Lambda_{L_{n}|V}(\theta)$ is minimized at
$\theta^{\ast}$ and equals 0 at $\theta=0$. Hence, the following relation
\begin{equation}
1_{\left \{  L_{n}>nb,V=\frac{v}{\phi(1-f_{n})}\right \}  }\hat{L}%
\leq1_{\left \{  L_{n}>nb,V=\frac{v}{\phi(1-f_{n})}\right \}  }\qquad \text{a.s.}
\label{indicatorineq}%
\end{equation}
holds for any $v$.

To prove the theorem, now we re-express
\begin{align*}
\mathbb{E}^{\ast}\left[  1_{\{L_{n}>nb\}}L^{\ast^{2}}\right]   &
=\mathbb{E}^{\ast}\left[  1_{\left \{  L_{n}>nb,V\leq \frac{v_{\delta}^{\ast}%
}{\phi(1-f_{n})}\right \}  }L^{\ast^{2}}\right]  +\mathbb{E}^{\ast}\left[
1_{\left \{  L_{n}>nb,V>\frac{v_{\delta}^{\ast}}{\phi(1-f_{n})}\right \}
}L^{\ast^{2}}\right] \\
&  =K_{1}+K_{2},
\end{align*}
where $v_{\delta}^{\ast}$ is the unique solution to the equation
$r(v)=b-\delta$.

The remaining part of proof will be divided into three steps.

\noindent \textbf{Step 1.} In this step, we show
\[
K_{1}=o(f_{n}).
\]
By Lemma \ref{leratiobound}, for sufficiently large $n$, there exists a finite
positive constant $C$ such that
\[
\frac{f_{V}(v)}{f_{V}^{\ast}(v)}\leq C\left(  -\log \phi(1-f_{n})\right)
\]
for all $v$. From \eqref{indicatorineq}, it then follows that
\[
1_{\left \{  L_{n}>nb,V\leq \frac{v_{\delta}^{\ast}}{\phi(1-f_{n})}\right \}
}L^{\ast^{2}}\leq C\left(  -\log \phi(1-f_{n})\right)  \left(  1_{\left \{
L_{n}>nb,V\leq \frac{v_{\delta}^{\ast}}{\phi(1-f_{n})}\right \}  }L^{\ast
}\right)  \qquad \text{a.s.}%
\]
Therefore, $K_{1}$ is upper bounded by
\begin{align*}
\mathbb{E}^{\ast}\left[  1_{\left \{  L_{n}>nb,V\leq \frac{v_{\delta}^{\ast}%
}{\phi(1-f_{n})}\right \}  }L^{\ast^{2}}\right]   &  \leq C\left(  -\log
\phi(1-f_{n})\right)  \left(  \mathbb{E}^{\ast}\left[  1_{\left \{
L_{n}>nb,V\leq \frac{v_{\delta}^{\ast}}{\phi(1-f_{n})}\right \}  }L^{\ast
}\right]  \right) \\
&  =C\left(  -\log \phi(1-f_{n})\right)  \left(  \mathbb{P}\left(
L_{n}>nb,V\leq \frac{v_{\delta}^{\ast}}{\phi(1-f_{n})}\right)  \right) \\
&  \leq C\left(  -\log \phi(1-f_{n})\right)  \exp(-\beta n).
\end{align*}
The last step is due to step 1 in the proof of Theorem \ref{th4.1}. Moreover,
by Lemma \ref{lephifn}, $-\log \phi(1-f_{n})\sim \alpha \log \left(  \frac
{1}{f_{n}}\right)  =o\left(  \frac{1}{f_{n}}\right)  $. Note $f_{n}$ has a
sub-exponential decay rate, it implies $\frac{1}{f_{n}}\exp(-\beta
n/2)\rightarrow0$. Therefore, $K_{1}$ is still $o(f_{n})$.

\noindent \textbf{Step 2.} We show that
\begin{equation}
\limsup_{n\rightarrow \infty}\frac{\log K_{2}}{\log f_{n}}\leq2.
\label{k2bound}%
\end{equation}
By Jensen's inequality,
\begin{align*}
\mathbb{E}^{\ast}\left[  1_{\left \{  L_{n}>nb,V>\frac{v_{\delta}^{\ast}}%
{\phi(1-f_{n})}\right \}  }L^{\ast^{2}}\right]   &  \geq \left(  \mathbb{E}%
^{\ast}\left[  1_{\left \{  L_{n}>nb,V>\frac{v_{\delta}^{\ast}}{\phi(1-f_{n}%
)}\right \}  }L^{\ast}\right]  \right)  ^{2}\\
&  =\left(  \mathbb{P}\left(  L_{n}>nb,V>\frac{v_{\delta}^{\ast}}{\phi
(1-f_{n})}\right)  \right)  ^{2}\\
&  \sim f_{n}^{2}\left(  \frac{(v^{\ast})^{-1/\alpha}}{\Gamma(1-1/\alpha
)}\right)  ^{2},
\end{align*}
where the last step is due to Theorem \ref{th4.1}. Then \eqref{k2bound}
follows by applying the logarithm function on both sides and using the fact
that $\log \left(  f_{n}\right)  <0$ for all sufficiently large $n$.

\noindent \textbf{Step 3.} We show that
\begin{equation}
\liminf_{n\rightarrow \infty}\frac{\log K_{2}}{\log f_{n}}\geq2.
\label{k3bound}%
\end{equation}
First note that, on the set $\left \{  L_{n}>nb,V>\frac{v_{\delta}^{\ast}}%
{\phi(1-f_{n})}\right \}  $, by \eqref{indicatorineq} the likelihood ratio
$L^{\ast}$ is upper bounded by $\frac{f_{V}(v)}{f_{V}^{\ast}(v)}$ and hence by
\eqref{ratiovdelta}, with sufficiently large $n$, it holds for all
$v>\frac{v_{\delta}^{\ast}}{\phi(1-f_{n})}$ that
\begin{align*}
\frac{f_{V}(v)}{f_{V}^{\ast}(v)}  &  <\frac{C_{0}}{\overline{F}_{V}(x_{0}%
)}x_{0}^{1/\log \phi(1-f_{n})}\left(  -\log \phi(1-f_{n})\right)  v^{-1/\alpha
-\frac{1}{\log \phi(1-f_{n})}+\varepsilon}\\
&  \leq C\left(  -\log \phi(1-f_{n})\right)  \left(  \frac{v_{\delta}^{\ast}%
}{\phi(1-f_{n})}\right)  ^{-1/\alpha-\frac{1}{\log \phi(1-f_{n})}+\varepsilon
}\\
&  <C\left(  -\log \phi(1-f_{n})\right)  \left(  \phi(1-f_{n})\right)
^{1/\alpha+\frac{1}{\log \phi(1-f_{n})}-\varepsilon}.
\end{align*}
Multiplying it with the indicator and taking expectation under $\mathbb{E}%
^{\ast}$, we obtain
\[
\mathbb{E}^{\ast}\left[  1_{\left \{  L_{n}>nb,V>\frac{v_{\delta}^{\ast}}%
{\phi(1-f_{n})}\right \}  }L^{\ast^{2}}\right]  \leq C^{2}\left(  -\log
\phi(1-f_{n})\right)  ^{2}\left(  \phi(1-f_{n})\right)  ^{2/\alpha+\frac
{2}{\log \phi(1-f_{n})}-2\varepsilon}.
\]
Then, taking logarithms on both sides, dividing by $\log f_{n}$ and by Lemma
\ref{lephifn}, we obtain
\[
\liminf_{n\rightarrow \infty}\frac{\log \mathbb{E}^{\ast}\left[  1_{\left \{
L_{n}>nb,V>\frac{v_{\delta}^{\ast}}{\phi(1-f_{n})}\right \}  }L^{\ast^{2}%
}\right]  }{\log f_{n}}\geq2-2\alpha \varepsilon.
\]
Finally, \eqref{k3bound} is yield by letting $\varepsilon \downarrow0$.

Combining Step 1, Step 2 and Step 3, the desired result asserted in the
theorem is obtained.\bigskip
\end{proof}

The following two proofs are motivated by \cite{chan2010efficient}. Lemma
\ref{leorder} below will be used in proving Lemma \ref{thboundederror}.

\begin{lemma}
\label{leorder} Let $R_{1},\ldots,R_{n}$ be an i.i.d. sequence of standard
exponential random variables. Suppose $R_{(k)}$ is the $k$th order statistic
and $\lim_{n\to \infty}\frac{k}{n}=a<1$. Then, for every $\varepsilon>0$, there
exists a constant $\beta>0$ such that the following inequality
\[
\mathbb{P}\left(  \left \vert R_{(k)}-\log \left(  \frac{1}{1-a}\right)
\right \vert \geq \varepsilon \right)  \leq \frac{\beta}{n}.
\]
holds for all sufficiently large $n$.
\end{lemma}

\begin{proof}
For i.i.d. standard exponential random variables $R_{i},i=1,\ldots,n$, it
follows from \cite{renyi1953theory} that
\[
R_{(k)}\overset{d}{=}\sum_{j=1}^{k}\frac{R_{j}}{n-j+1}.
\]
Then,
\begin{equation}
\mathbb{E}[R_{(k)}]=\sum_{j=1}^{k}\frac{1}{n-j+1}=H_{n}-H_{n-k}\to \log \left(
\frac{1}{1-a}\right)  , \quad \text{as }n\to \infty, \label{expectrk}%
\end{equation}
where $H_{n}$ denotes the $n$th harmonic number, i.e., $H_{n}=1+\frac{1}%
{2}+\cdots+\frac{1}{n}$ for $n\geq1$. \eqref{expectrk} is verified by noting
the following asymptotic expansion; see, e.g., \cite{berndt1998ramanujan},
\[
H_{n}\sim \log(n)+\gamma+O\left(  \frac{1}{n}\right)  ,
\]
and $\gamma$ is the Euler's constant. Similarly,
\begin{equation}
\mathrm{Var}[R_{(k)}]=\sum_{j=1}^{k}\left(  \frac{1}{n-j+1}\right)  ^{2}%
=H_{n}^{(2)}-H_{n-k}^{(2)}\sim \frac{a}{1-a}\frac{1}{n}, \quad \text{as }%
n\to \infty, \label{varrk}%
\end{equation}
where $H_{n}^{(2)}$ is the $n$th harmonic number of order 2, i.e.,
$H_{n}^{(2)}=1+\frac{1}{2^{2}}+\cdots+\frac{1}{n^{2}}$ for $n\geq1$.
\eqref{varrk} is derived by applying the asymptotic expansion of $H_{n}^{(2)}%
$; see, e.g., \cite{berndt1998ramanujan},
\[
H_{n}^{(2)}\sim \frac{\pi^{2}}{6}-\frac{1}{n}+O\left(  \frac{1}{n^{2}}\right)
.
\]
Then, by Chebyshev's inequality, it follows that, for every $n>0$,
\[
\mathbb{P}\left(  \vert R_{(k)}-\mathbb{E}[R_{(k)}]\vert \geq \varepsilon
\right)  \leq \frac{\mathrm{Var}[R_{(k)}]}{\varepsilon^{2}}.
\]
Due to \eqref{expectrk} and \eqref{varrk}, there exists $N$, such that for all
$n\geq N$,
\[
\mathbb{P}\left(  \left \vert R_{(k)}-\log \left(  \frac{1}{1-a}\right)
\right \vert \geq \varepsilon \right)  \leq \frac{\beta}{n},
\]
where $\beta$ only depends on $\varepsilon$ and $a$. \bigskip
\end{proof}

\begin{proof}
[Proof of Lemma \ref{thboundederror}]Recall that $O_{i}=\frac{R_{i}}%
{\phi(1-l_{i}f_{n})}$, for all $i=1,\ldots,n$. Then the order statistic
$O_{(k)}$ is almost surely lower bounded by
\[
\frac{R_{(k)}}{\phi \left(  1-\max \limits_{j\leq|\mathcal{W}|}l_{j}%
f_{n}\right)  }.
\]
Since $k=\min \{l:\sum_{i=1}^{l}c_{(i)}>nb\}$, we have
\[
\liminf_{n\rightarrow \infty}\frac{k}{n}\geq \frac{b}{\max \limits_{j\leq
|\mathcal{W}|}c_{j}}:=b^{\prime}.
\]
Fix $\varepsilon>0$. For all sufficiently large $n$, $\mathbb{E}\left[
S^{2}(\mathbf{R})\right]  $ can be bounded as follows,
\begin{align*}
\mathbb{E}\left[  S^{2}(\mathbf{R})\right]   &  \leq \mathbb{E}\left[
\mathbb{P}\left(  V>\frac{R_{(\lfloor nb^{\prime}\rfloor)}}{\phi \left(
1-\max \limits_{j\leq|\mathcal{W}|}l_{j}f_{n})\right)  }\right)  ^{2}\right] \\
&  \leq \mathbb{E}\left[  \left(  \mathbb{P}\left(  V>\frac{R_{(\lfloor
nb^{\prime}\rfloor)}}{\phi \left(  1-\max \limits_{j\leq|\mathcal{W}|}l_{j}%
f_{n})\right)  },R_{(\lfloor nb^{\prime}\rfloor)}\geq \log \left(  \frac
{1}{1-b^{\prime}}\right)  -\varepsilon \right)  \right.  \right. \\
&  +\left.  \left.  \mathbb{P}\left(  V>\frac{R_{(\lfloor nb^{\prime}\rfloor
)}}{\phi \left(  1-\max \limits_{j\leq|\mathcal{W}|}l_{j}f_{n})\right)
},R_{(\lfloor nb^{\prime}\rfloor)}<\log \left(  \frac{1}{1-b^{\prime}}\right)
-\varepsilon \right)  \right)  ^{2}\right] \\
&  \leq \left(  \mathbb{P}\left(  V>\frac{\log \left(  \frac{1}{1-b^{\prime}%
}\right)  -\varepsilon}{\phi \left(  1-\max \limits_{j\leq|\mathcal{W}|}%
l_{j}f_{n})\right)  }\right)  +\mathbb{P}\left(  R_{(\lfloor nb^{\prime
}\rfloor)}<\log \left(  \frac{1}{1-b^{\prime}}\right)  -\varepsilon \right)
\right)  ^{2}.
\end{align*}
Then,
\begin{align*}
\limsup_{n\rightarrow \infty}\frac{\mathbb{E}\left[  S^{2}(\mathbf{R})\right]
}{f_{n}^{2}}  &  \leq \left(  \limsup_{n\rightarrow \infty}\frac{\mathbb{P}%
\left(  V>\frac{\log \left(  \frac{1}{1-b^{\prime}}\right)  -\varepsilon}%
{\phi \left(  1-\max \limits_{j\leq|\mathcal{W}|}l_{j}f_{n})\right)  }\right)
}{f_{n}}+\limsup_{n\rightarrow \infty}\frac{\mathbb{P}\left(  R_{(\lfloor
nb^{\prime}\rfloor)}<\log \left(  \frac{1}{1-b^{\prime}}\right)  -\varepsilon
\right)  }{f_{n}}\right)  ^{2}\\
&  \leq \left(  \max \limits_{j\leq|\mathcal{W}|}l_{j}\frac{\left(  \log \left(
\frac{1}{1-b^{\prime}}\right)  -\varepsilon \right)  ^{-1/\alpha}}%
{\Gamma(1-1/\alpha)}+M\right)  ^{2}<\infty.
\end{align*}
The last step is due to the regular variation of $V$, Lemma \ref{leorder} and
the condition that $\frac{1}{n}=O(f_{n})$.
\end{proof}
\bigskip

\textbf{Acknowledgments}

We are grateful to the Editor and the anonymous reviewer for the helpful comments and
suggestions that have greatly improved the presentation of the paper. Ken Seng Tan acknowledges the research funding from the Society of Actuaries CAE’s grant and the Singapore University Grant. Fan Yang acknowledges financial
support from the Natural Sciences and Engineering Research Council of Canada
(grant number: 04242).

\bibliographystyle{apalike}
\bibliography{reference_1}

\end{document}